\documentclass[11pt]{article}
\usepackage{amssymb, authblk}
\usepackage{amsmath, bbm}
\usepackage{fullpage}
\usepackage{amsthm, natbib, hyperref}
\usepackage{graphicx}
\usepackage{verbatim}
\usepackage{float}

\usepackage{resizegather}

\usepackage{hyperref}
\usepackage{cleveref}

\usepackage[title]{appendix}

\usepackage{threeparttable}
\usepackage{algorithm}
\usepackage{algpseudocode}
\usepackage[dvipsnames]{xcolor}
\usepackage{tikz}
\usetikzlibrary{arrows.meta}
\usetikzlibrary{decorations.pathreplacing}

\algnewcommand\algorithmicinput{\textbf{INPUT:}}
\algnewcommand\INPUT{\item[\algorithmicinput]}
\algnewcommand\algorithmicoutput{\textbf{OUTPUT:}}
\algnewcommand\OUTPUT{\item[\algorithmicoutput]}

\DeclareMathAlphabet{\mathpzc}{OT1}{pzc}{m}{it}

\usepackage{hyperref}[]
\hypersetup{
    colorlinks=true,
    linkcolor=blue,
    filecolor=magenta,      
    urlcolor=cyan,
      citecolor=blue,
    }

\usepackage{cleveref}

\newtheorem{theorem}{Theorem}

\allowdisplaybreaks

\date{\vspace{-5ex}}

\usepackage{tikz}
\newcommand{\Mypm}{\mathbin{\tikz [x=1.4ex,y=1.4ex,line width=.1ex] \draw (0.0,0) -- (1.0,0) (0.5,0.08) -- (0.5,0.92) (0.0,0.5) -- (1.0,0.5);}}

\usepackage{booktabs,tabularx,graphicx}
\newcolumntype{C}{>{\centering\arraybackslash}X}
\usepackage[margin=2.5cm]{geometry}

\usepackage{xspace}

\author{Marcos Matabuena$^{1}$ and Oscar Hernan Madrid Padilla$^{2}$ \\
		$^{1}$CiTIUS (Centro Singular de Investigaci\'{o}n en Tecnolox\'{i}as Intelixentes) \\  $^{2}$Department of Statistics, University of California Los Angeles \\
	$^{1}$\url{marcos.matabuena@usc.es}, $^{2}$\url{oscar.madrid@stat.ucla.edu}}

\title{Energy distance and kernel mean embeddings for two-sample survival testing}

\date{\today}

\begin{document}
\maketitle
\begin{abstract}
We study the comparison problem  of  distribution equality between two random samples under a right censoring scheme. To address this problem,  we design a series of tests based on  energy distance and kernel mean embeddings. We calibrate our tests using permutation methods and prove that they are consistent against all fixed continuous
alternatives. To evaluate our proposed  tests, we simulate survival curves from previous clinical trials. Additionally, we provide practitioners with a set of recommendations on how to select parameters/distances for the delay effect problem.  Based on the method for parameter tunning that we propose, we show that our tests demonstrate a considerable gain of statistical power against classical survival tests. 
\end{abstract}

\section{Introduction}
 One of the main objectives of survival analysis is to compare the distributions of  the lifetime of two populations. This is best illustrated by means of clinical trials when evaluating the efficacy of two treatments \cite{biblsingh2011survival}. In the context of right censored data, the scientific community tends to use the log-rank test  to testing the equality between two distribution curves. Originally  proposed by \cite{mantel1959statistical}, the log-rank test  has further been studied  by different  authors, e.g.,  \cite{schoenfeld1981asymptotic,fleming2011counting}. Importantly,  the   log-rank test is known to be the most powerful test when the hazard functions are proportional to each other (\cite{schoenfeld1981asymptotic}). However, when this hypothesis is violated, the test has a significant loss of power \cite{fleming1980modified,lachin1986evaluation,lakatos1988sample,schoenfeld1981asymptotic}.


  Currently, an important area  of statistical  research is searching for new tests that guarantee high statistical power in real use-cases  where log-rank test does not perform well. We refer the reader  to  \cite{su2018time} where  the authors thoroughly  discuss a lack of statistical power of the log-rank test found in numerous  case studies. Recent cancer immunotherapy trials also provide a relevant example. These  consist of situations 
   where  treatments  may present a delayed effect \cite{melero2014therapeutic,xu2017designing,xu2018designing,su2018time,alexander2018hazards}.
  
  
In the right censoring survival setting, we  distinguish two  different  types of tests: directionals and  omnibus. Loosely speaking, the former seek to obtain maximum power in specific scenarios, while the latter are consistent against all alternatives. Examples of directional tests are  the log-rank test family, see e.g  \cite{gehan1965generalized,tarone1977distribution,peto1972asymptotically,fleming1981class}, where  statistics are assigned a weight function that determines the optimality in certain directions.  Other approaches include combinations of  tests, such  as those in \cite{bathke2009combined} and \cite{yang2010improved}.

From a theoretical point of view, omnibus tests are often preferred over directionals due to their ability to detect any alternative asymptotically. However, in practice, these tests have the disadvantage  that they may   have low local power  versus a wide variety of alternatives. In addition, it is known
 that any test with finite samples can have high power only in a limited number of scenarios.  In particular, \cite{janssen2000global} proves that there exists no test with high power, except in a finite dimensional space. 


In the era of precision medicine, see \cite{kosorok2019precision} for a review,  drugs are designed to be personalized. This makes  the statistical  analysis  of  treatment  differences particularly  challenging. For example,  a comparison of two treatments in a group of individuals may present highly heterogeneous survival curves due to significant individual variability in response to the treatment. A particular instance of this can be seen in  immunotherapy studies \cite{ferris2016nivolumab}  (Figure $1$, Image $B$), where  the survival curves intersect several times.

In this paper,  we propose a novel approach for the  two-sample  testing problem under right  censoring. Our approach  relies on   energy distance
 (\citealt{szekely2003statistics} and \citealt{szekely2013energy}) and  maximum mean discrepancy estimation (\citealt{gretton2012kernel}).  We summarize  our  contributions next.



\subsection{Summary of results}
\label{sec:summary}

Formally,  we consider  the  classical traditional framework of two-sample  survival comparisons where we are given  lifetimes $T_{j,i} \sim P_j$ $(j=0,1; i=1,\dots,n_j)$ and censoring times $C_{j,i} \sim$ $Q_j$ $(j=0,1; i=1,\dots,n_j)$, with distributions   $P_j$ and $Q_j$  $(j=0,1)$, defined in a subset of $\mathbb{R^{+}}$. Here,  the index   $j$  represents  a population, and the index  $i$  a particular  sample within a   population. Moreover, the random variables  $T_{0,1},\dots,T_{0,n_0},\dots,T_{1,1},\dots,T_{1,n_1},$ $ C_{0,1},\dots,C_{0,n_0},\dots,C_{1,1},\dots,C_{1,n_1}$ are assumed to be independent of each other. In practice, only the random variables  $X_{j,i}= min(T_{j,i},C_{j,i})$  and $\delta_{j,i}= 1\{X_{j,i} = T_{j,i}\}$	$(j=0,1; i=1,\dots,n_j)$ are observed. 

On the basis  of the  observed  data $\{(X_{j,i},\delta_{j,i})\}_{j=0,1; i=1,\dots,n_j}$,  the two-sample testing  problem that we study can be formulated as
\begin{equation}
\label{eqn:model}
	 H_0\,:\, P_0(t)  =  P_1(t), \,\,\forall  t> 0,\,\,\,\,\,\,\,\text{versus}\,\,\,\,\,\,\,  H_A\,:\, P_0(t)  \neq   P_1(t), \,\,\text{for some}\,\,\, t> 0.
\end{equation}

Our main contributions are the following:

\begin{itemize}
	\item We propose  novel  tests  based  on   energy distance and maximun mean discrepancy. The resulting  tests require minimum assumptions, involving only conditions on the moments of random variables. Specifically,  we  assume $ E(T_{j,i}^{2})<\infty $ and $ E(C_{j,i}^ {2}) <\infty$, and, for simplicity, that the variables $X_{j,i}, T_{j,i},C_{j,i}$ $(j=0,1; i=1,\dots,n_j)$ are continuous. 

	\item Importantly,  we show   that  the proposed  tests  are  consistent  against all alternatives. In addition, we   present  a permutation-based  procedure to  approximate the distribution of our test statistics under the null hypothesis.
	
	\item We provide guidance on how to tune parameters of our proposed  tests  in clinical situations of interest. Furthermore, we show that Gaussian and Laplacian kernels  outperform energy distance with euclidean distance  and other tests of the logrank family, in settings where there is  a delay effect, a commonly  found situation in contemporary clinical trials.

\end{itemize}

Finally, we extend the proposed method to the multivariate case (appendix \ref{appendix:D}) and  demonstrate the theoretical properties of the proposed  statistics  (appendix \ref{appendix:B}). In particular, we show that these statistics behave as  true distances between  samples.



\subsection{Outline}


 The structure of the paper is as follows. Section \ref{sec:background}  provides an introduction to energy distance-based methods. Next, in Section \ref{estadisticos}
  the statistics for our tests are derived, establishing their connections  with previous  work on two-sample  testing  based on kernel methods. Subsequently, we propose a permutation method and some recommendations on how to choose the tests parameters. In Section \ref{sec:theory},  we show that our tests are consistent  against all alternatives. Section \ref{sec:simulation} then provides a simulation study to compare the behavior of the proposed  tests against state-of the-art methods. To this end, we  compare the type $I$ error using known distributions. In addition, we consider real scenarios from clinical practice and  evaluate performance based  on  the power of the tests.  Finally, the validity of our methods is  verified  in practice using the previously collected data     (\cite{stablein1981analysis}).

 In order to increase readability of the present document, we place the proofs of the main theoretical contributions and complementary results in the appendices.

\section{Background on energy  distance}
\label{sec:background}

To arrive at our  family of tests,  we first  recall some background on energy distance. To that end, let  $X$,$X^{\prime}\sim^{i.i.d. } P$ and $Y$,$Y^{\prime}\sim^{i.i.d.}Q$  where  
$P$ and  $Q$ are probability  distribution functions in $\mathbb{R}^{d}$. Denoting  by $\|\cdot\| $ the Euclidean distance in $\mathbb{R}^d$ and assuming that $\max\{E(||X||), E(||Y||)\}< \infty$, the  energy distance  between the distributions $P$ and $Q$ is defined, as in \cite{szekely2003statistics} and \cite{szekely2013energy},  by:
\begin{equation}
\label{eqn:def1}
\epsilon(P,Q)= 2E||X-Y||- E||X-X^{'}||-E||Y-Y^{'}||.
\end{equation}\\
It is fairly easy to see that  $\epsilon(\cdot, \cdot )$  is invariant to rotations, non-negative, and $\epsilon(P,Q)= 0$ if and only if $P=  Q$. In addition, 	
(\ref{eqn:def1}) can be extended for a family of parameters $\alpha\in(0,2]$ assuming in each case the existence of the moment of order $\alpha$, see \cite{szekely2013energy}. The corresponding  $\alpha$-energy distance is then given as
\begin{equation}
\label{eqn:def2}
\epsilon_\alpha(P,Q)= 2E||X-Y||^{\alpha}- E||X-X^{'}||^{\alpha}-E||Y-Y^{'}||^{\alpha}.
\end{equation}

It can be proved that $\epsilon_\alpha(P,Q)\geq 0$. Furthermore, $\epsilon_{\alpha}(P,Q)=0$ if and only if $P=Q$. In the case of $\alpha=2$, $\epsilon_2(P,Q)= 2||E(X)-E(Y)||^{2}$. Therefore, non-negativity is verified trivially, although 	$\epsilon_2(P,Q)=0$ implies equality in means and not that $P = Q$.

For  a  characteristic  kernel  $K \,:\, \mathbb{R}^{d} \times \mathbb{R}^{d}  \rightarrow \mathbb{R}$  using properties of kernel mean embeddings \cite{muandet2017kernel}, as in \cite{gretton2012kernel}, we define the  measure of maximum mean discrepancy (MMD) as 
\begin{equation}
	\label{eqn:mdd}
	\gamma_K^{2}(P,Q)= E(K(X,X^{'}))+E(K(Y,Y^{'}))-2E(K(X,Y)),
\end{equation}
	where $X$,$X^{\prime}\sim^{\text{i.i.d.}} P$ and $Y$,$Y^{\prime}\sim^{\text{i.i.d.}} Q$.  Intuitively, (\ref{eqn:mdd}) can be thought of  as  non-linear generalization of the energy distance (\ref{eqn:def1}) in an appropriate  reproducing  kernel Hilbert space (RKHS).  The latter  depends on the selected parameters/distances.

Following this line, if we consider the  energy distance in metric spaces \cite{lyons2013distance} (with an arbitrary semi-metric of negative type instead of the Euclidean distance), we find it equivalent to the kernel methods just defined. This equivalence was established in \cite{sejdinovic2013equivalence} and \cite{shen2018exact}, at both the population and sample level.

	Finally, some typical  examples  of characteristics kernel 	\cite{sriperumbudur2010hilbert}  are provided in Table \ref{tab:tabla1}.

	\begin{table}[t!]
		\begin{center}
			\caption{	\label{tab:tabla1}	\textbf{Characteristics kernels}.
				$\Gamma(\cdot)$ denotes the Gamma function and $K_v$ is the modified Bessel
				function of the second order
				$v$ (see explicit definitions in the Appendix \ref{appendix:F})}
			\label{tab:table1}
			\begin{tabular}{c|c} 
				\textbf{Kernel Function} & \textbf{$K(x,y)$} \\ 
				\hline
				Gaussian & $\exp((\frac{||x-y||}{\sigma})^2), \sigma>0$ \\ 
				Laplacian &  $\exp(\frac{|x-y|}{\sigma}),   \sigma>0$ \\
				Rational quadratic & $(||x-y||+c)^{-\beta}$, $\beta,\alpha>0$ \\
				Mattern   & $\frac{2^{1-v}}{\Gamma(v)}(\frac{\sqrt{2v}||x-y||}{\sigma})K_v(\frac{\sqrt{2v}||x-y||}{\sigma})$ \\
				
			\end{tabular}
		\end{center}
	\end{table}


\section{Methodology}
\label{estadisticos}

In this  section,  we present  a new family of tests which are the   focus of this paper. We begin by providing  constructions of the statistics  that are the pillars of our tests. Then, we present  a procedure  for determining the  distribution of the statistics under the null  hypothesis.

\subsection{Construction of statistics}
\label{tiempos}

In the context of right censoring with independent  data, the maximum non-parametric 
likelihood approach is the Kaplan-Meier  estimator  originally introduced in \cite{kaplan1958nonparametric}. Notably, the Kaplan-Meier  estimator is consistent (\cite{wang1987note}) and  its  asymptotic  properties were studied in  \cite{cai1998asymptotic}.  However, \cite{stute1994bias}  showed that the Kaplan-Meier estimator  suffered from negative  bias, which can be  large under high censoring.

To proceed with our  construction,  we  exploit  the Kaplain-Meier  estimator,  combining it  with a  kernel type of estimator  based on energy distance. To this end,
for each group $j \in \{0,1\}$, we consider its ordered sample $$X_{j,(1:n_j)}  <  X_{j,(2:n_j)}<\dots <X_{j,(n_j:n_j)},$$
and  the corresponding censored indicators $\delta_{j,(1:n_j)}$, $\delta_{j,(2:n_j)},\dots \delta_{j,(n_j:n_j)}$. In addition, we refer to the maximum possible lifetimes  for each group as  $\tau_0$ and $\tau_1$ respectively.


With the above notation  in hand, we motivate  the   definition  of our statistics. First,  if we knew the distributions $P_0$ and $P_1$, then we could  calculate the metrics defined in (\ref{eqn:def2})  or (\ref{eqn:mdd})  to measure  the  distance  between the two populations. Since these distributions are not available, it is then natural to  estimate   them with the Kaplan-Meier estimator and use a  sample version  of the distances    (\ref{eqn:def2})  or (\ref{eqn:mdd}). This leads to 
an   energy distance statistic under right censoring:
\begin{gather}\label{sin normalizarr}
\begin{align}
\tilde{\epsilon}_{\alpha}(P_0,P_1)= 2\sum_{i=1}^{n_0}  \sum_{j=1}^{n_1} W^{0}_{i:n_0} W^{1}_{j:n_1} ||X_{0,(i:n_0)}-X_{1,(j:n_1)}||^{\alpha} - \sum_{i=1}^{n_0} \sum_{j=1}^{n_0}  W^{0}_{i:n_0} W^{0}_{j:n_0}  ||X_{0,(i:n_0)}-X_{0,(j:n_0)}||^{\alpha} \\ - \sum_{i=1}^{n_1}  \sum_{j=1}^{n_1} W^{1}_{i:n_1} W^{1}_{j:n_1} ||X_{1,(i:n_1)}-X_{1,(j:n_1)}||^{\alpha} \nonumber,
\end{align}
\end{gather}
and a kernel  statistic under right censoring:
\begin{gather}
\begin{align}
\tilde{\gamma}_K^{2}(P_0,P_1)=  \sum_{i=1}^{n_0}  \sum_{j=1}^{n_0} W^{0}_{i:n_0} W^{0}_{j:n_0} K(X_{0,(i:n_0)},X_{0,(j:n_0)})+ \sum_{i=1}^{n_1}  \sum_{j=1 }^{n_1} W^{1}_{i:n_1} W^{1}_{j:n_1} K(X_{1,(i:n_1)},X_{1,(j:n_j)})\\-2\sum_{i=1}^{n_0}  \sum_{j=1}^{n_1} W^{0}_{i:n_0} W^{1}_{i:n_1} K(X_{0,(i:n_0)},X_{1,(j:n_1)}),\nonumber
\end{align}
\end{gather}

where

\begin{equation}
W^{0}_{{i:n_0}}= \frac{\delta_{0,(i:n_0)}}{n_0-i+1}\prod_{j=1}^{i-1}[\frac{n_0-j}{n_0-j+1}]^{\delta_{0,(j:n_0)}}
\hspace{0.2cm}   (i=1,\dots,n_0),
\end{equation}

and

\begin{equation}
W^{1}_{{i:n_1}}= \frac{\delta_{1,(i:n_1)}}{n_1-i+1}\prod_{j=1}^{i-1}[\frac{n_1-j}{n_1-j+1}]^{\delta_{1,(j:n_1)}} \hspace{0.2cm}   (i=1,\dots,n_1)
\end{equation}are the Kaplan-Meier  weights from \cite{stute1995statistical}. While the statistics  $\tilde{\epsilon}_{\alpha}(P_0,P_1)$ and  $\tilde{\gamma}_K^{2}(P_0,P_1)$ seem to capture the differences between two populations, it is possible to prove that, almost surely, $\tilde{\epsilon}_{\alpha}(P_0,P_1)$ and  $\tilde{\gamma}_K^{2}(P_0,P_1)$  converge to quantities $\gamma_{c(K)}(P_0,P_1)$ and $\epsilon_{c(\alpha)}(P_0,P_1)$, respectively. However, they do not behave like distances between probability distributions. Specifically,  there exist two  different  probability  distributions $P_0$ and $P_1$    in $\mathbb{R} $  satisfying  $\epsilon_{c(\alpha)}(P_0, P_1)  <0$. We can also find  two  different  probability  distributions $P_0$ and $P_1$    in $\mathbb{R} $ with  $\epsilon_{c(\alpha)}(P_0, P_1)=0$. We refer the reader to the appendix \ref{appendix:C} for  specific  constructions of these examples.

The reason  behind the odd behavior of the statistic $\tilde{\gamma}_K^{2}(P_0,P_1)$  ($\tilde{\epsilon}_{\alpha}(P_0,P_1)$) has to do with the fact that   $P_l$ is not completely supported in $[0,\tau_l]$, for $l  \in \{0,1\}$. We alleviate this  problem by defining  the conditional distributions  $P^{\prime}_0 (x)= P_0(x)/ \int_{0}^{\tau_0} dP_0(x)dx$ $\forall x\in [0,\tau_0]$, and  $P^{\prime}_1 (x)= P_1 (x)/ \int_{0}^{\tau_1} dP_1 (x)dx $ $\forall x\in [0,\tau_1]$. With $P_0^{\prime}$ and $P_1{^\prime}$ at hand,  we  construct   conditional  versions of the weights $W^{l}_{i:n_l}$ $(l=0,1; i=1,\dots,n_j)$. Specifically, we consider the $U$-statistics under right censoring suggested in \cite {bose1999strong} and apply the aforementioned standardization, following \cite{stute1993multi}. The resulting statistics are:

\begin{gather}\label{eqn:estadistico1}
\begin{align}
\tilde{\epsilon}_{\alpha}(P_0,P_1)= 2  \frac{\sum_{i=1}^{n_0}  \sum_{j=1}^{n_1} W^{0}_{i:n_0} W^{1}_{i:n_1} ||X_{0,(i:n_0)}-X_{1,(j:n_1)}||^{\alpha}}{\sum_{i=1}^{n_0}  \sum_{j=1}^{n_1} W^{0}_{i:n_0} W^{1}_{j:n_1}} - \frac{\sum_{i=1}^{n_0} \sum_{j\neq i}^{n_0}  W^{0}_{i:n_0} W^{0}_{j:n_0}  ||X_{0,(i:n_0)}-X_{0,(j:n_0)}||^{\alpha}}{\sum_{i=1}^{n_0} \sum_{j\neq i}^{n_0} W^{0}_{i:n_0} W^{0}_{j:n_0}} \\
- \frac{\sum_{i=1}^{n_1}  \sum_{i\neq j}^{n_1} W^{1}_{i:n_1} W^{1}_{j:n_1}  ||X_{1,(i:n_1)}-X_{1,(j:n_1)}||^{\alpha}}{\sum_{i=1}^{n_1}  \sum_{j\neq i}^{n_1} W^{1}_{i:n_1} W^{1}_{j:n_1}} \nonumber
\end{align}
\end{gather}
\textbf{($U$-statistic $\alpha$-energy distance under right censoring)},

\begin{gather}\label{eqn:estadistico2}
\begin{align}
\tilde{\gamma}_K^{2}(P_0,P_1)=  \frac{\sum_{i=1}^{n_0}  \sum_{j\neq i}^{n_0} W^{0}_{i:n_0} W^{0}_{j:n_0} K(X_{0,(i:n_0)},X_{0,(j:n_0)})}{ \sum_{j=1}^{n_0} \sum_{j\neq i}^{n_0} W^{0}_{i:n_0} W^{0}_{j:n_0}}+ \frac{\sum_{i=1}^{n_1}  \sum_{j\neq i}^{n_1} W^{1}_{i:n_1} W^{1}_{j:n_1}  K(X_{1,(i:n_1)},X_{1,(j:n_1)})}{\sum_{i=1}^{n_1}  \sum_{j\neq i }^{n_1}W^{1}_{i:n_1} W^{1}_{j:n_1}}\\-2\frac{\sum_{i=1}^{n_0}  \sum_{j=1}^{n_1} W^{0}_{i:n_0} W^{1}_{i:n_1} K(X_{0,(i:n_0)},X_{1,(j:n_1)})}{\sum_{i=1}^{n_0}  \sum_{j=1}^{n_1} W^{0}_{i:n_0} W^{1}_{i:n_1}} \nonumber
\end{align}
\end{gather}

\textbf{($U$-statistic kernel method under right censoring)}.

Analogously, we can define V-statistics in the following manner:

\begin{gather}\label{eqn:estadistico3}
\begin{align}
\tilde{\epsilon}_{\alpha}(P_0,P_1)= 2  \frac{\sum_{i=1}^{n_0}  \sum_{j=1}^{n_1} W^{0}_{i:n_0} W^{1}_{j:n_1} ||X_{0,(i:n_0)}-X_{1,(j:n_1)}||^{\alpha}}{\sum_{i=1}^{n_0}  \sum_{j=1}^{n_1} W^{0}_{i:n_0} W^{1}_{j:n_1}} - \frac{\sum_{i=1}^{n_0} \sum_{j=1 }^{n_0}  W^{0}_{i:n_0} W^{0}_{j:n_0}  ||X_{0,(i:n_0)}-X_{0,(j:n_0)}||^{\alpha}}{\sum_{i=1}^{n_0} \sum_{j= 1}^{n_0} W^{0}_{i:n_0} W^{0}_{j:n_0}} \\
- \frac{\sum_{i=1}^{n_1}  \sum_{j=1}^{n_1} W^{1}_{i:n_1} W^{1}_{j:n_1}  ||X_{1,(i:n_1)}-X_{1,(j:n_1)}||^{\alpha}}{\sum_{i=1}^{n_1}  \sum_{j=1}^{n_1} W^{1}_{i:n_1} W^{1}_{j:n_1}} \nonumber
\end{align}
\end{gather}

\textbf{($V$-statistic $\alpha$-energy distance under right censoring)},

\begin{gather}\label{eqn:estadistico4}
\begin{align}
\tilde{\gamma}_K^{2}(P_0,P_1)=  \frac{\sum_{i=1}^{n_0}  \sum_{j=1}^{n_0} W^{0}_{i:n_0} W^{0}_{j:n_0} K(X_{0,(i:n_0)},X_{0,(j:n_0)})}{ \sum_{j=1}^{n_0} \sum_{j=1}^{n_0} W^{0}_{i:n_0} W^{0}_{j:n_0}}+ \frac{\sum_{i=1}^{n_1}  \sum_{j=1}^{n_1} W^{1}_{i:n_1} W^{1}_{j:n_1}  K(X_{1,(i:n_1)},X_{1,(j:n_1)})}{\sum_{i=1}^{n_1}  \sum_{j=1 }^{n_1}W^{1}_{i:n_1} W^{1}_{j:n_1}}\\-2\frac{\sum_{i=1}^{n_0}  \sum_{j=1}^{n_1} W^{0}_{i:n_0} W^{1}_{i:n_1} K(X_{0,(i:n_0)},X_{1,(j:n_1)})}{\sum_{i=1}^{n_0}  \sum_{j=1}^{n_1} W^{0}_{i:n_0} W^{1}_{i:n_1}} \nonumber
\end{align}
\end{gather}

\textbf{($V$-statistic kernel method under right censoring)}.

Finally, to establish the consistency more easily, our final statistics are given as

\begin{equation}
T_{\tilde{\epsilon}_\alpha}= \frac{n_0n_1}{n_0+n_1} \tilde{\epsilon}_{\alpha}(P_0,P_1) \text{   and   }    T_{\tilde{\gamma}_K^{2}}= \frac{n_0n_1}{n_0+n_1} \tilde{\gamma}_K^{2}(P_0,P_1). 
\end{equation}


In Appendix \ref{appendix:C}, we can find,  in some instances, an interpretation of the limits of these statistics. In particular, we show that the statistics behave as distances between  distribution functions and the  characteristic functions  in a weighted Hilbert space $L^{2}(I)$.

\subsection{Permutation tests}

As in the case of  the usual energy  two-sample test from  \cite{szekely2003statistics}  and \cite{szekely2013energy},
 the null distribution of  our proposed statistics is approximated with a  permutation method. If the censorship mechanism  of the two groups is the same,  the standard permutation method    from \cite{neuhaus1993conditional} and  \cite{wang2010testing} is valid. However, when the censoring distributions differ, the standard permutation method does not perform well in small-sample settings or when the amount of censoring is large, see \cite{heimann1998permutational}. In this case, one alternative is to use the re-sampling strategy proposed in  \cite{wang2010testing}. Below we describe the steps of the classical permutation procedure.


We denote by $Z=(\overbrace{0,\cdots,0}^{n_0},\overbrace{1,\cdots,1}^{n_1})$ a vector of size $n=n_0+n_1$ that  indicates  the observed  group  membership. Thus,  $z_i = 1$ ($z_i = 0$) indicates   that the $i$-th  subject  belongs to group $1$  ($0$).  We then order the observed  times and  censorship indicators, thus we  construct   vectors
 $U= (X_{0,1},\cdots, X_{0,n_0},X_{1,1},\cdots ,X_{1,n_1})$ and $\delta=(\delta_{0,1},\cdots ,\delta_{0,n_0},\delta_{1,1},\cdots ,\delta_{1,n_1})$. Next,  if we are interested in calculating the  distribution of the statistic  $\theta (Z, U,\delta)$ under the null distribution ($P_0 = P_1$), then we can proceed to construct permutations of the  data. Specifically, let $\mathcal{S}$ be  a collection  of sets of size  $n_0$  whose elements  belong  to $\{1,\ldots, n_0 + n_1\}$. For  every  $I \in \mathcal{S}$, we  construct  a vector $Z^{I} \in \mathbb{R}^n$
 satisfying $Z_i^I = 0 $  if  $i \in I$ and  $Z_i^I  =  1$  if  $i \notin I$. Next, we compare $\theta(Z,U,\delta)$  against $\theta(Z^I,U,\delta)$  for all $I \in \mathcal{S}$. The p-value is calculated  as
\begin{equation}
\label{eqn:pvalue}
\text{p-value}=\frac{\sum_{ I \in  \mathcal{S} }1\{{\theta(Z^I,U,\delta)\geq \theta(Z,U,\delta)\}}}{\binom{n}{n_0}}.
\end{equation}

In practice, we can reduce  the number of operations in (\ref{eqn:pvalue}) by using a random  subset $\mathcal{S}^{\prime}$  of $\mathcal{S}$ to  obtain
\begin{equation*}
\text{p-value}\approx \frac{\sum_{ I \in  \mathcal{S}^{\prime} }1\{{\theta(Z^I,U,\delta)\geq \theta(Z,U,\delta)\}}}{  \vert \mathcal{S}^{\prime}\vert }.
\end{equation*}

\subsection{Selection of tuning parameters/distances}\label{heuristic}

Although the proposed methods are consistent against all alternatives from an asymptotic point of view, (see  Theorem \ref{thm1}), one of the main practical difficulties with finite samples is the selection of parameters/distances so that high statistical power is guaranteed. In fact, this problem is very common in kernel methods both in prediction models and hypothesis testing.  \cite{filippi2016bayesian} state that there exist few theoretical approaches to tackle this problem.

In this work, we  only use the energy distance with the Euclidean distance and the Gaussian and Laplacian kernels (see Table \ref{tab:table1}). The main reason for this is that there is a corpus of previous work on how the selection of parameters influences  the performance of different methods. There are also  some heuristics that include  theoretical results, see  \cite{ramdas2015decreasing} and \cite{garreau2017large}.



Despite  the  fact that  energy  distance is more  sensitive to the choice of the $\alpha$ than rather that to that of kernel (see for example \cite{sejdinovic2013equivalence}), there is no known formal criterion for selecting an optimal value of $\alpha$.


In regard to the Gaussian and Laplacian kernels, there is a known add-hoc rule called Median heuristic that consists in selecting the median between the distance pairs of the aggregate sample. This procedure is explained in detail below.

Let $X= (X_1,\dots,X_{n_0},X_{n_0+1},\cdots,X_{n_0+n_1})= (X_{0,1},\cdots,X_{0,n_0}, X_{1,n_1} ,\cdots, X_{1,n_1})$ be the aggregate sample vector. Consider $D  \in \mathbb{R}^{  (n_0+n_1) \times  (n_0+n_1) }$ defined as $D_{ij}=  |X_i-X_j| (i=1,\dots ,(n_0+n_1),j=1,\dots ,(n_0+n_1))$.

As in \cite{garreau2017large},  we define

$$\sigma= \sqrt{H_{n}/2} \hspace{0.2cm},\,\,\,\,\, \text{where} \,\,\,\,\, \hspace{0.2cm} H_{n}= median\{D_{ij}^{2}:1\leq i<j\leq (n_0+n_1)\}. $$

In the literature, the resulting $\sigma$  is known  as kernel bandwidth.  An intuitive explanation of how this works is given below:

\begin{itemize}
	\item Given $X_i,X_j$ ($i=1,\dots ,(n_0+n_1),j=1,\dots ,(n_0+n_1)$), if $\sigma\to 0$ or $\sigma\to \infty$, then $K(X_i,X_j)\to1$ or $K(X_i,X_j)\to 0$  (see Table \ref{tab:tabla1}).  Therefore,  $\tilde{\gamma}_K^{2}(P_0,P_1)$ is almost always constant (see equation \ref{eqn:estadistico2}), and  the statistical power  of the test is  low.   
	\item It is reasonable to impose that the median of  $D_{ij}$ ($i=1,\dots ,(n_0+n_1),j=1,\dots ,(n_0+n_1)$) and $\sigma$  are of the same order so that $K(X_i,X_j)$ ($i=1,\dots ,(n_0+n_1),j=1,\dots ,(n_0+n_1)$) does not take   unnecessarily small or large values, so as not to suffer from the limitations mentioned above. 
	\item Hence, a reasonable choice for $\sigma$  is “in the middle range” of  $D_{ij}$ ($i=1,\dots ,(n_0+n_1),j=1,\dots ,(n_0+n_1)$). In this way, $\sigma$ is of the same order as median of $D_{ij}$ ($i=1,\dots ,(n_0+n_1),j=1,\dots ,(n_0+n_1)$). The global dispersion between terms $K(X_i,X_j)$ ($i=1,\dots ,(n_0+n_1),j=1,\dots ,(n_0+n_1)$) is maximized, and therefore, the test has greater discrimination capacity.

\end{itemize}

Alternatively, $\sigma$ is sometimes set to $ \sqrt{H_n}$.
The influence of the suboptimal specification of the kernel bandwidth has mainly been  studied  in situations of high dimensionality. In this context, it has been shown to lead to important differences in  power of  tests. For instance, \cite{ramdas2015decreasing}  noticed, using a simulation study and theoretical analysis,  that the median heuristic $\sigma$   maximized power with Gaussian kernel in several cases. However, power can be suboptimal with the Laplacian kernel, showing better results with some values of $\sigma=H_n^{\alpha}$ for $\alpha\in (0,2]$ with $\alpha\neq 1/2$. In any case, we should be cautious interpreting these results.  As we do not consider the multidimensional case, the effects of a suboptimal kernel bandwidth specification may not be so dramatic in our setting.

In the case of censorship, in addition to the vector  $X$, we also have to consider the  vector $\delta= (\delta_{0,1},\cdots,\delta_{0,n_0}, \delta_{1,n_1} ,\cdots, \delta_{1,n_1})$ with censorship indicators. Now, we define the set of indices $I=\{i\in\{1,2,\dots,(n_0+n_1)\}: \delta_i= 1 \}$. A reasonable estimator for $\sigma$  is given by  $\sigma= \sqrt{H_n^{*}}$ or $\sigma= \sqrt{H_n^{*}/2}$ where 
 
 $$H^{*}_{n}= median\{D_{ij}^{2}:1\leq i<j\leq (n_0+n_1) \hspace{0.2cm} \text{with} \hspace{0.2cm} i,j\in I\}.$$

 The previous definition is justified because   in equations  (\ref{eqn:estadistico1})--(\ref{eqn:estadistico4}), only the elements whose indices belong to $I$ influence  the corresponding expressions.

\section{Theory}
\label{sec:theory}

Next, we show that,  under very mild conditions, our proposed  tests are consistent  against  all alternatives. This is formally stated  below and  the proof  can be found   in Appendix \ref{appendix:proof}.

\begin{theorem}
	\label{thm1}
	Let  $X_{j,i}= min(T_{j,i},C_{j,i})\sim^{ \text{i.i.d.}  } P_{c(j)} $ and  $\delta_{j,i}= 1\{X_{j,i}= T_{j,i}\}$	$(j=0,1; i=1,\dots,n_j)$ with $P_{c(j)}$ $(j=0,1)$. Suppose also that  the conditions  stated in Section \ref{sec:summary} hold  for the random variables $T_{j,i}\sim^{ \text{i.i.d.}  }  P_{j},\,\,\,\,\,C_{j,i} \sim^{ \text{i.i.d.}  }  Q_{j}$ $(j=0,1; i=1,\dots,n_j)$. Further assume that $\tau_0=\tau_1$  or the  support of the distribution functions $P_0$ and $P_1$ is  contained in the intervals $ [0,\tau_0] $ and $[0,\tau_1]$, respectively. Then, 
for testing the null  $H_0  \,:\, P_{0}(t)  =  P_{1}(t) \,\,\,\,\,\forall  t\in [0,\tau_1] $	
	 the statistics $T_{\tilde{\epsilon}_\alpha}$ and   $T_{\tilde{\gamma}_K^{2}}$  determine tests that are consistent against all fixed  alternatives with continuous random variables.
\end{theorem}

	The supposition that $\tau_0=\tau_1$ can be modified by truncating the random variables in $[0,\tau]$, with $\tau= \min\{\tau_0,\tau_1\}$. This way,  it is guaranteed that the  $T_{\tilde{\epsilon}_\alpha}$ and   $T_{\tilde{\gamma}_K^{2}}$ statistics define consistent  tests against all the alternatives.  \cite{schumacher1984two} followed the same approach with the Kolmogorov-Smirnov and Cramer Von-Mises tests under censorship. Note that one of the two hypotheses must be verified because  we can have two distribution functions that take the same values in $[0,\tau]$. However they can differ in $(\tau, \infty)$ and the statistics  in the limit take the same value by the  normalization' that we use in the statistics (\ref{eqn:estadistico1})--(\ref{eqn:estadistico4}). Naturally, if the end right of supports of the distribution functions $P_{0}$ and $P_{1}$ are different, and  are contained in $[0,\tau_0]$ and $[0,\tau_1]$ respectively,   the test will show differences even if $P_0$ and $P_1$ take the same values in $[0,\tau]$.
		
	

The Kolmogorov-Smirnov and Cramer Von-Mises tests under censorship have been proposed in the context of absolutely continuous random variables \cite{schumacher1984two}. Unlike those tests, our results are also valid for discrete distributions provided that the second-order moments of the random variables exist. This can be very important in practice, since many of the lifetimes collected in databases for simplification are truncated and discrete in nature (see for example \cite{cai2019lce} and  \url{http://lce.biohpc.swmed.edu/lungcancer/dataset.php}).

	


\section{Simulation study}
\label{sec:simulation}

The simulation study is divided in two phases. In the first part, we consider scenarios where the null hypothesis is true. Then, the performance of the proposed tests is compared with the log-rank family tests with different censorship rates and different sample size.  In particular, the tests used are the energy distance ($\alpha=1$), Gaussian kernel $(\sigma=1$), Laplacian kernel $(\sigma=1$),
log-rank (\cite{mantel1959statistical}), Gehan generalized Wilcoxon test (\cite{gehan1965generalized}), Tarone-Ware (\cite{tarone1977distribution}), Peto-Peto (\cite{peto1972asymptotically}), Fleming $\&$ Harrington (\cite{fleming1981class}) (with $\rho=\gamma=1$). For this purpose, parametric distributions such as normal exponential or lognormal are used.  


In the second phase, the same tests are compared in scenarios where the null hypothesis is not true.   As in \cite{guyot2012enhanced}, we use the Digitizeit software (\url{https://www.digitizeit.de/}) to extract several survival curves  from different clinical trials in which there was a delay effect, or  there was no clear violation of the hypothesis that the hazard functions were not maintained. Survival curves were extracted from the studies analyzed in the following two papers:  \cite{su2018time} and \cite{alexander2018hazards}. We also consider simulations under the hypothesis that the hazard functions are proportional. This is  to assess the power loss of our tests compared to the log-rank tests. In all comparisons, the $\sigma$ parameter of the Gaussian and Laplacian kernels (see Table \ref{tab:tabla1}) is selected with the methodology defined in the Section  \ref{heuristic}. 

 When the null hypothesis is true, the sample size $n\in \{20,50\}$. Otherwise, $n\in \{20,50,100,200\}$. The censorship mechanism was the same within each simulation performed.

All the tests  are executed with the statistical software \textbf{\textsf{R}}. For the family of the log-rank test the coin package \cite{hothorn2008implementing} is  used while the new tests were implemented in \textbf{\textsf{C ++}} and integrated  in \textbf{\textsf{R}} with the ``Rcpp''  \cite{eddelbuettel2011rcpp} and ``Rcpp Armadillo''  libraries. In all  cases, the tests were calibrated by the permutation method, with $1000$ permutations executed.

\subsection{Null hyphotesis}
\label{sec:null}

We perform $500$ Monte Carlo simulations in which the null hypothesis is correct. The censoring rates are $10$ and $30$ percent and the sample size of $20$ and $50$ individuals. Since  p-values  are distributed uniformly ($\text{Uniform}(0,1))$ under the null hypothesis, the mean of the observed p-values obtained should be close to $0.5$, and the standard deviation close to $\sqrt{1/12}=0.2886751$. Similarly, approximately $5$ percent of the observations should have a value less than $0.05$. In the Appendix \ref{appendix:E} Tables (\ref{append:tab:tabla3})--(\ref{append:tab:tabla5}),  we can see the results of calculations of the mean and standard deviation for each test. In  Tables (\ref{append:tab:tabla6})--(\ref{append:tab:tabla8}),  the proportion of p-values is shown to be approximately  less  than or equal to $0.05$ for the same cases.

The results of the proposed tests under the null hypothesis are consistent and similar to those of the log-rank test family.
 Certain discrepancies with the theoretical values  are acceptable when doing the comparison with $500$ Monte Carlo simulations in $8$ different tests. In turn, the Kaplan-Meier estimator used in our models as well as in some of the log-rank family models presents a certain bias that is dependent on the censoring ratio, which produces small deviations under what is expected in a theoretical framework under the null hypothesis.

\subsection{Alternative hypothesis}

We perform $500$ Monte Carlo simulations in different situations where  the null hypothesis does not hold. We  differentiate two cases: $i)$ simulated data from survival curves extracted from clinical trials 
by means of the Digitizelt  and $ii)$ simulated data from an exponential distribution where the hypothesis that hazard ratio functions are proportional is true. The value $\alpha = 0.05$ is used as the cut-off for significance.

\subsubsection{Survival curves from clinical trials}

The curves extracted in this article for comparison are as follows:
Figure $1$-$A$ from \cite{borghaei2015nivolumab}, Figure $2$-$A$  from \cite{rodriguez2016phase}, Figure $2$-$B$ \cite{motzer2015nivolumab}, Figure $1$-$B$ from \cite{ferris2016nivolumab}, Figure $1$-$B$ from \cite{bellmunt2017pembrolizumab}, and  Figure $1$-$C$ in \cite{borghaei2015nivolumab}.

These articles were compiled from  \cite{su2018time} and \cite{alexander2018hazards} who assessed  the limitations of log-rank in many clinical situations or the problem of using summary measures to describe a survival curve. In addition, \cite{alexander2018hazards} focused on the field of immunotherapy where there was often a long-term delay effect on survival, which  motivated the recent development of new tests for this situation, e.g.,  \cite{xu2017designing} and \cite{xu2018designing}.

We use the presence of a clear delay effect in one of the treatments with respect to the other as criteria for selecting the survival curves. Additionally, a curve was selected in which hypothesis that the function is hazard are proportional is not violated with experimental data, Figure $1$-$A$ in \cite{rodriguez2016phase}. In most of the selected curves, the tests used in the original papers did not show statistically significant differences. 

The process of reconstructing each pair of curves is as follows:

\begin{enumerate}
	\item Extraction of the numerical values of the curves through the software Digitizelt.
	\item Reconstruction of  the curves from the numerical values  in the statistical software \textbf{\textsf{R}}.
	\item Truncation of the support of the curves to minimum right end of both curves, that is $\tau= \min\{\tau_0,\tau_1\}$, where $\tau_0$ is the right end of the first curve, and, analogously, $\tau_1$ for the second curve.    
	\item Smoothing curves with  cubic smoothing spline, as in \cite{Hast:Tibs:1990}. 
	\item Applying piecewise anti-isotonic linear regression so that the generated curves decrease, see \cite{BA03873521}. Subsequently,  data from the  estimated curves are simulated. The censorship variable is $C\sim^{} Uniform(0,\tau)$ where $\tau$ is the maximal value of support common to both curves by $3$.

\end{enumerate} 

	In Figures	\ref{fig:fig1}--\ref{fig:fig3}, 	
we can see the Kaplan-Meier curves after simulating data from the generated curves (with samples size of $10000$ individuals per population) along with an evaluation of power.

The results are discussed below:

\begin{itemize}
	\item In  Figures	\ref{fig:fig2} and	\ref{fig:fig3}, all the images reflect a situation of delay between the two treatments. In addition, almost all the patients die in the interval of time studied. In this situation, all our methods outperform  the log-rank family test studied, especially the tests based on the Laplacian and Gaussian kernels. 
	\item In  Figure \ref{fig:fig1} (left), there is a small delay much smoother than those discussed above. In addition, there is a significant fraction of patients who survive. In this situation, all the tests have  low power, even when the sample size is equal to $200$, but the Fleming $\&$ Harrington test works better than our proposals.   
	\item In Figure \ref{fig:fig1} (right), the situation where the hypothesis that hazard functions  do not seem violated, our tests have  low power. As expected, the  best method in this case is the log-rank, although it does not present  high power either. Graphically, it can be seen that the degree of discrepancy between both curves is low.
	
\end{itemize}

\begin{figure}[H]
	\centering
	\begin{gather}
	\begin{tabular}{cc}
		\includegraphics[width=8.5cm,height=9cm,keepaspectratio]{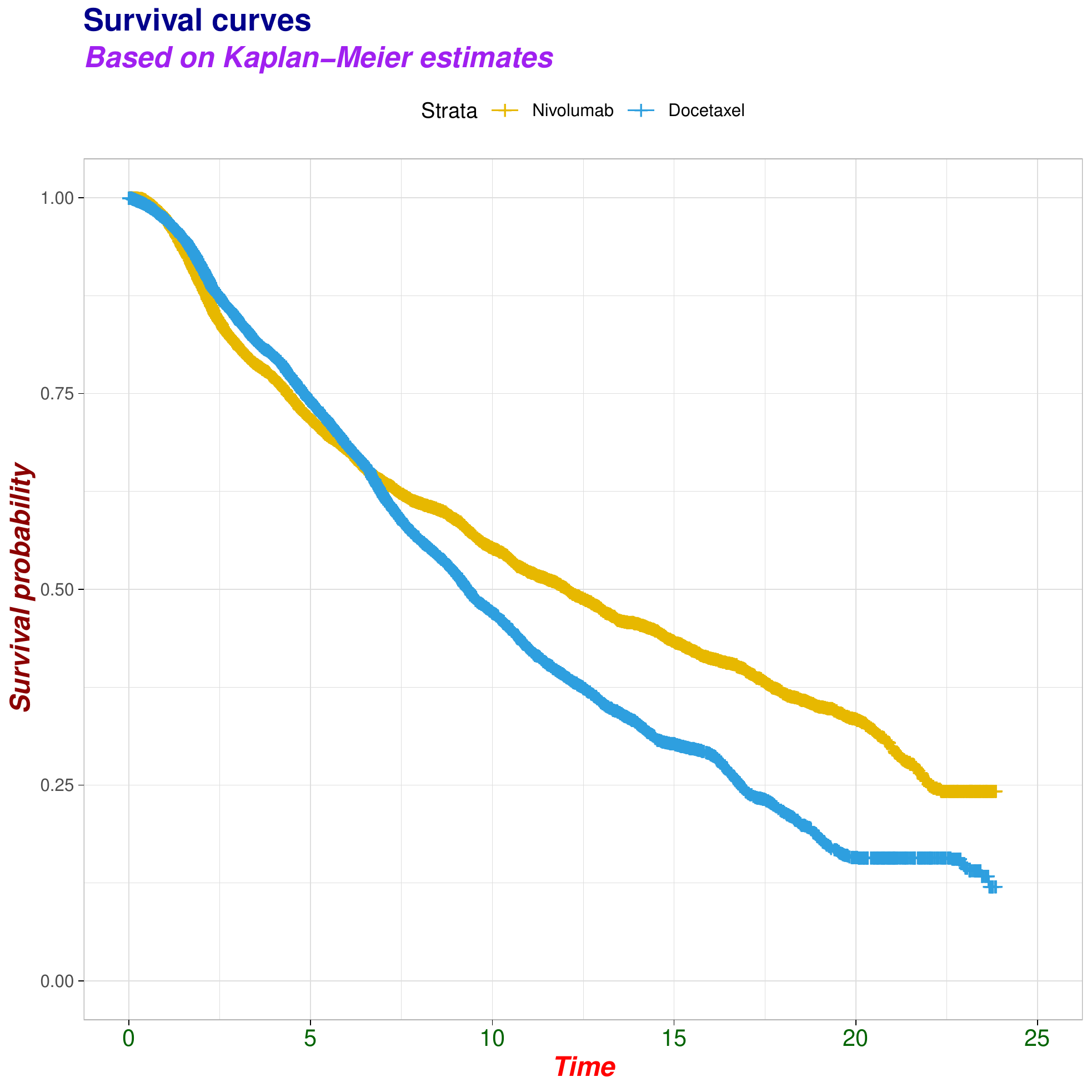}
		\includegraphics[width=8.5cm,height=9cm,keepaspectratio]{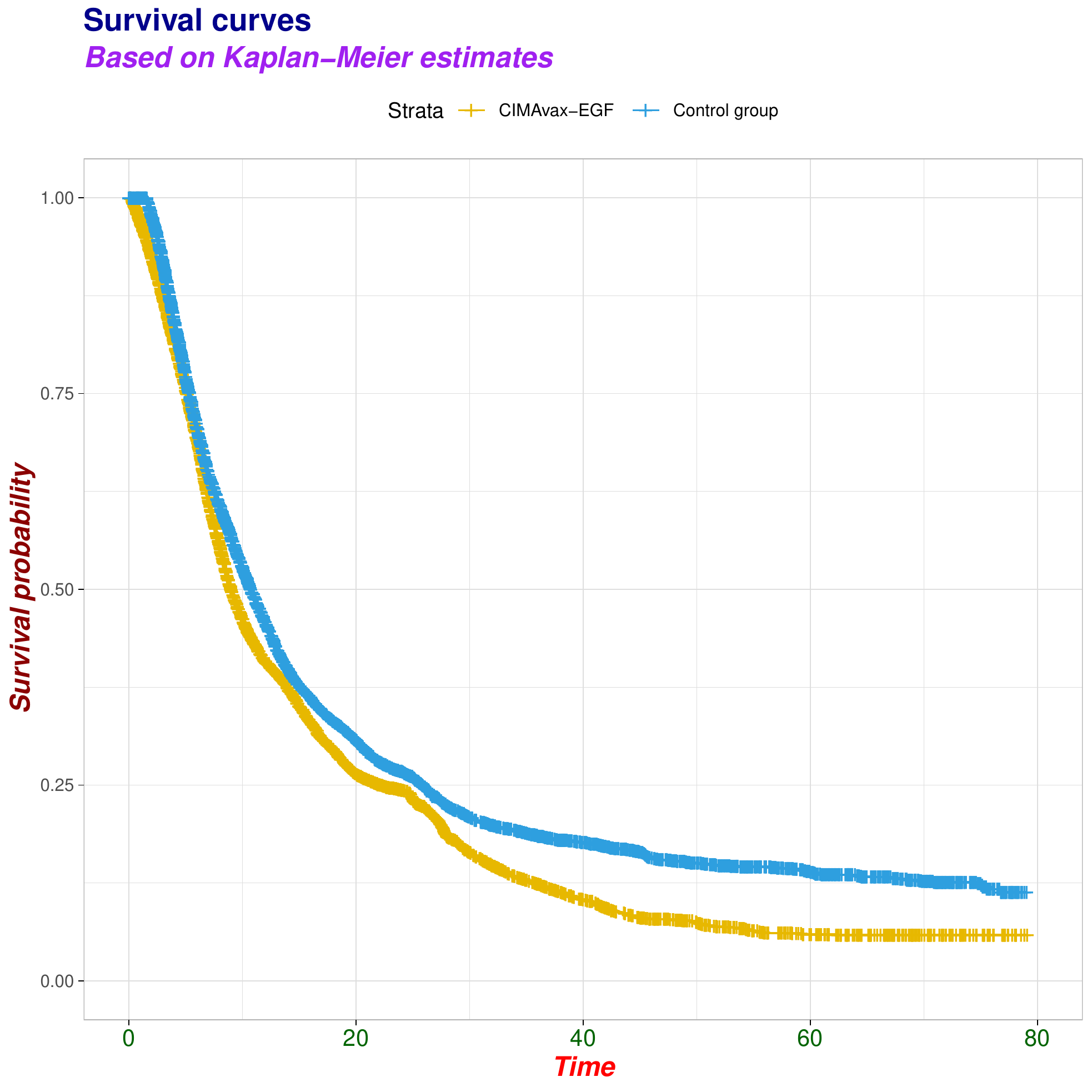}	
		\\
		\includegraphics[width=8.5cm,height=9cm,keepaspectratio]{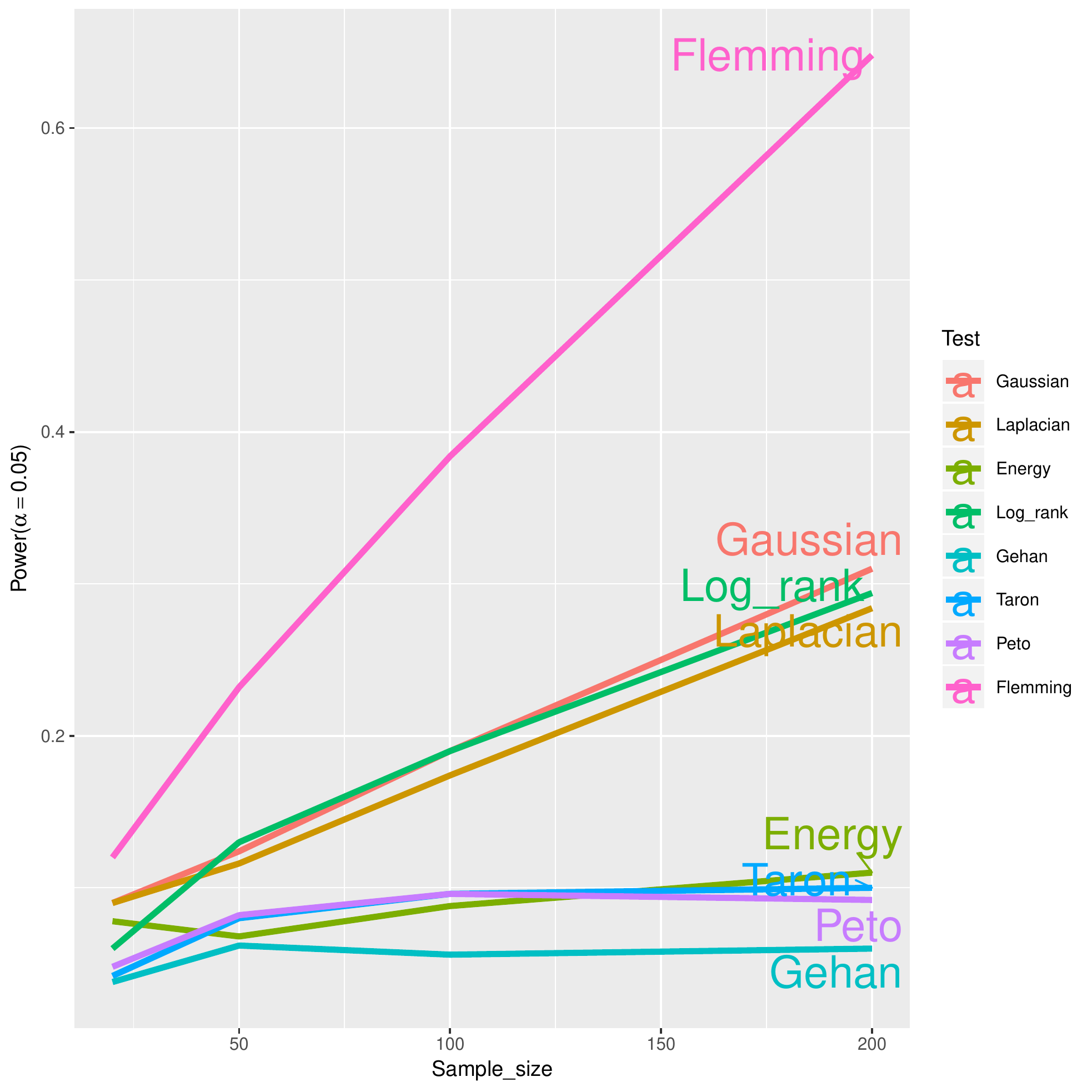}
		\includegraphics[width=8.5cm,height=9cm,keepaspectratio]{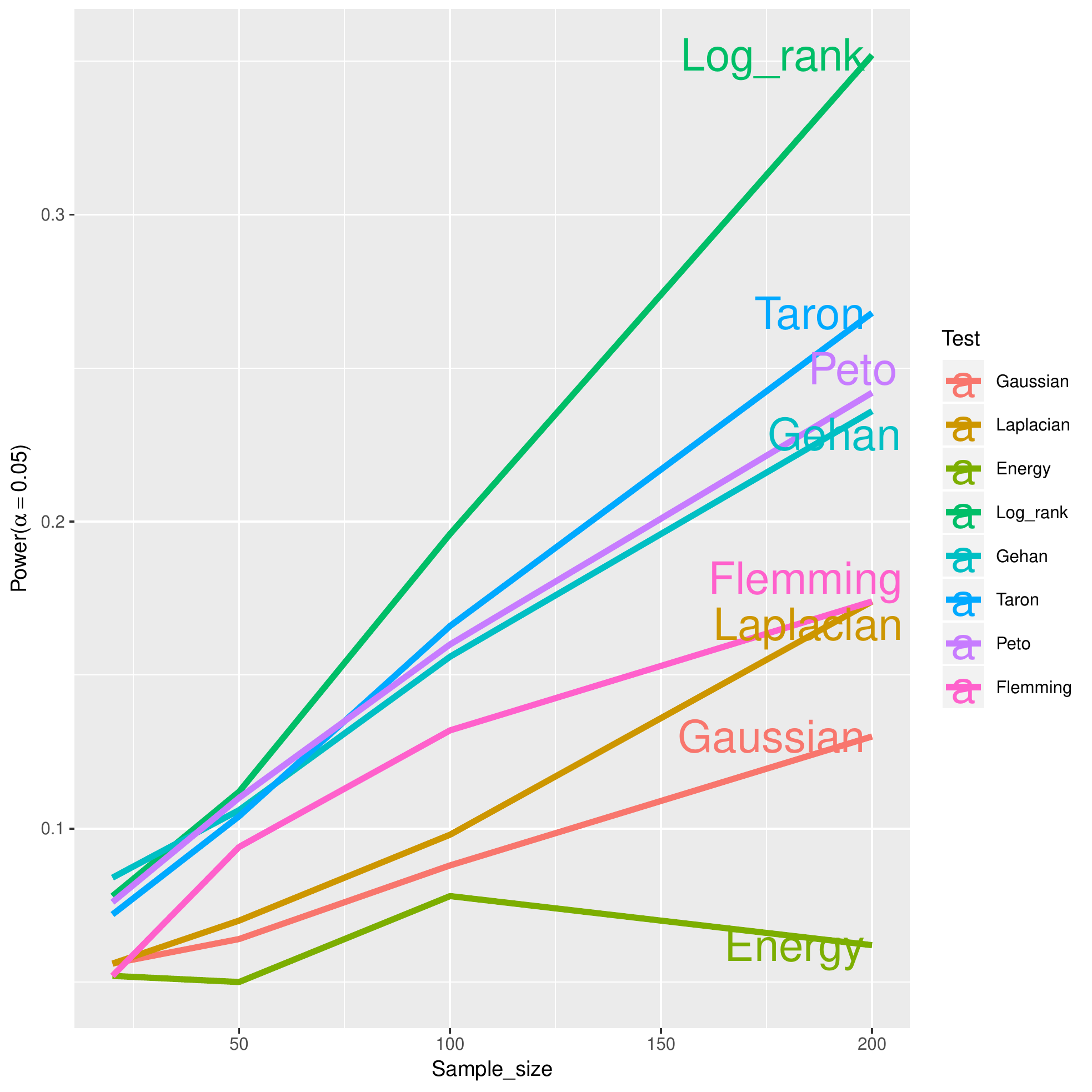}	
	\end{tabular}
	\end{gather}
	\caption{Statistical power of survival curves extracted \cite{borghaei2015nivolumab} Figure 1-$A$ (left) from and \cite{rodriguez2016phase} Figure 2-$A$ (right).} \label{fig:fig1}
\end{figure}

\begin{figure}[H]
	\centering
\begin{gather}
	\begin{tabular}{cc}
		\includegraphics[width=8.5cm,height=9cm,keepaspectratio]{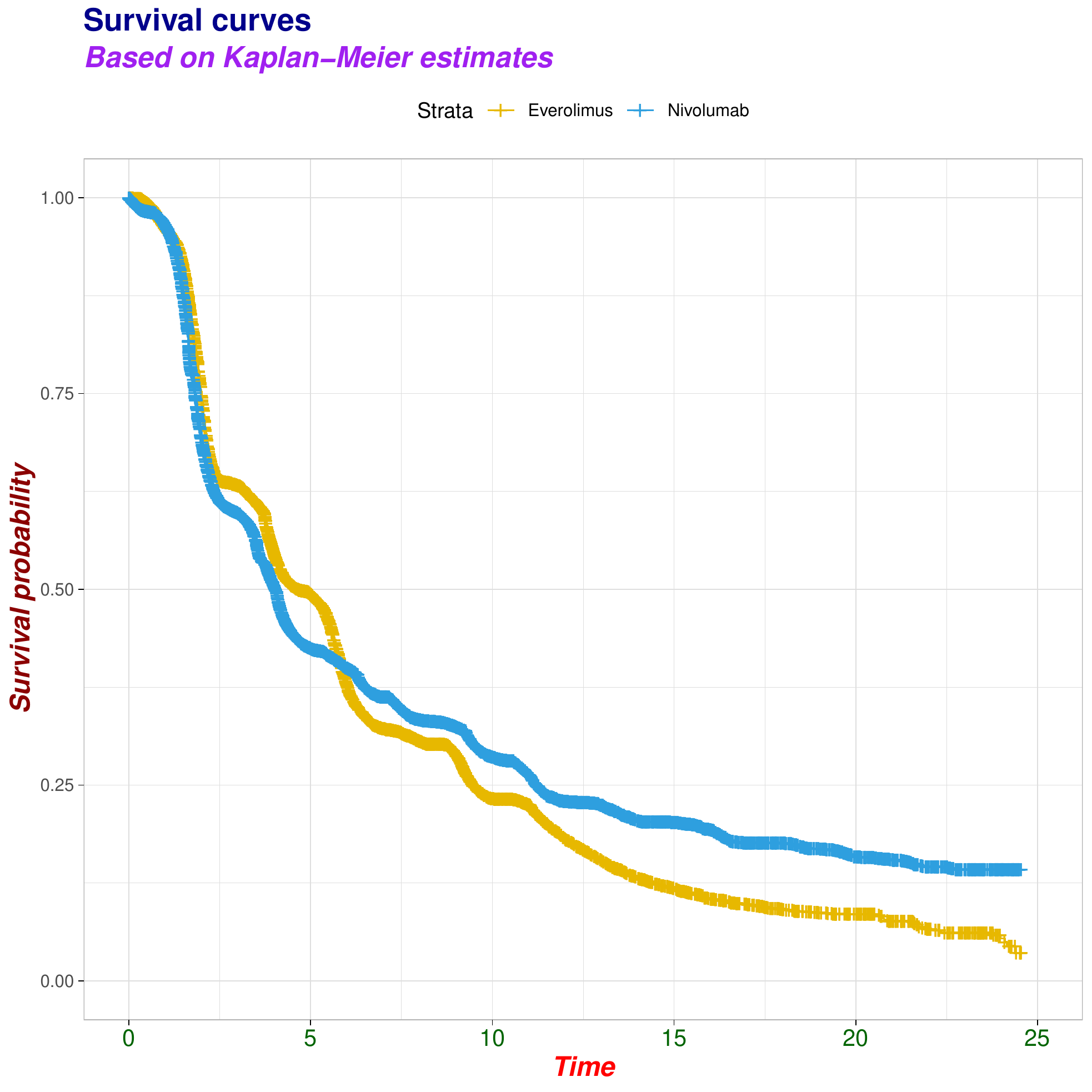}
		\includegraphics[width=8.5cm,height=9cm,keepaspectratio]{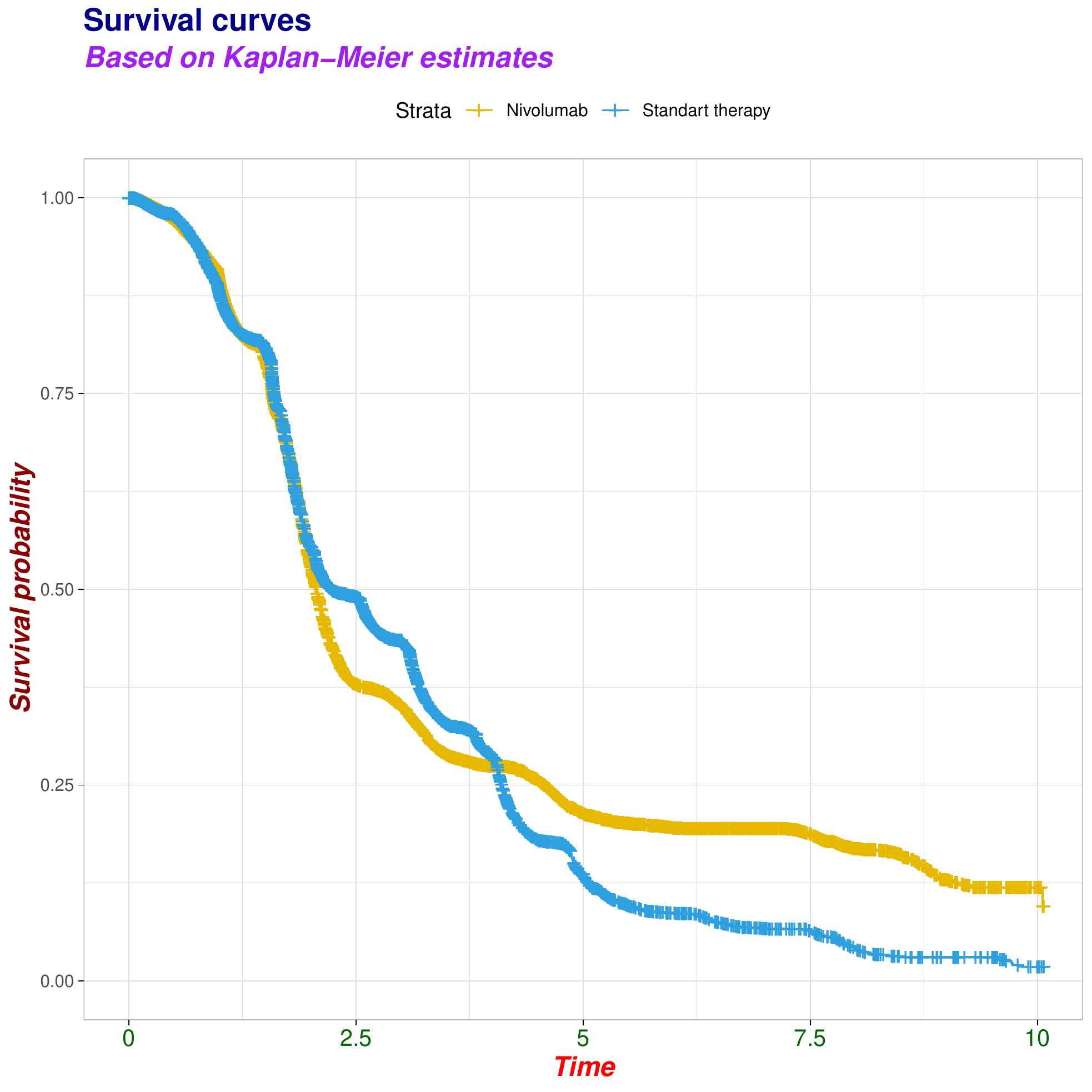}
		\\
		\includegraphics[width=8.5cm,height=9cm,keepaspectratio]{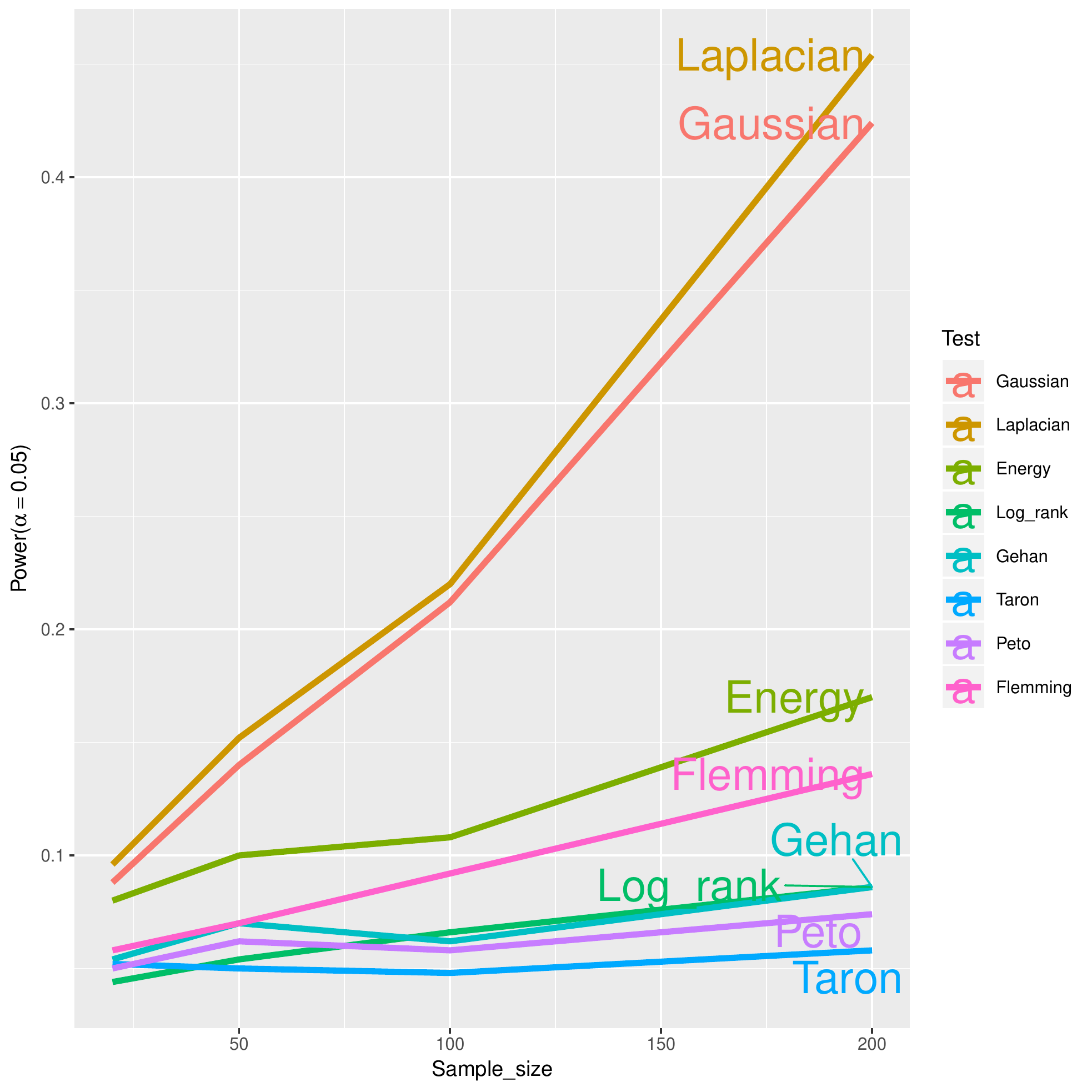}
		\includegraphics[width=8.5cm,height=9cm,keepaspectratio]{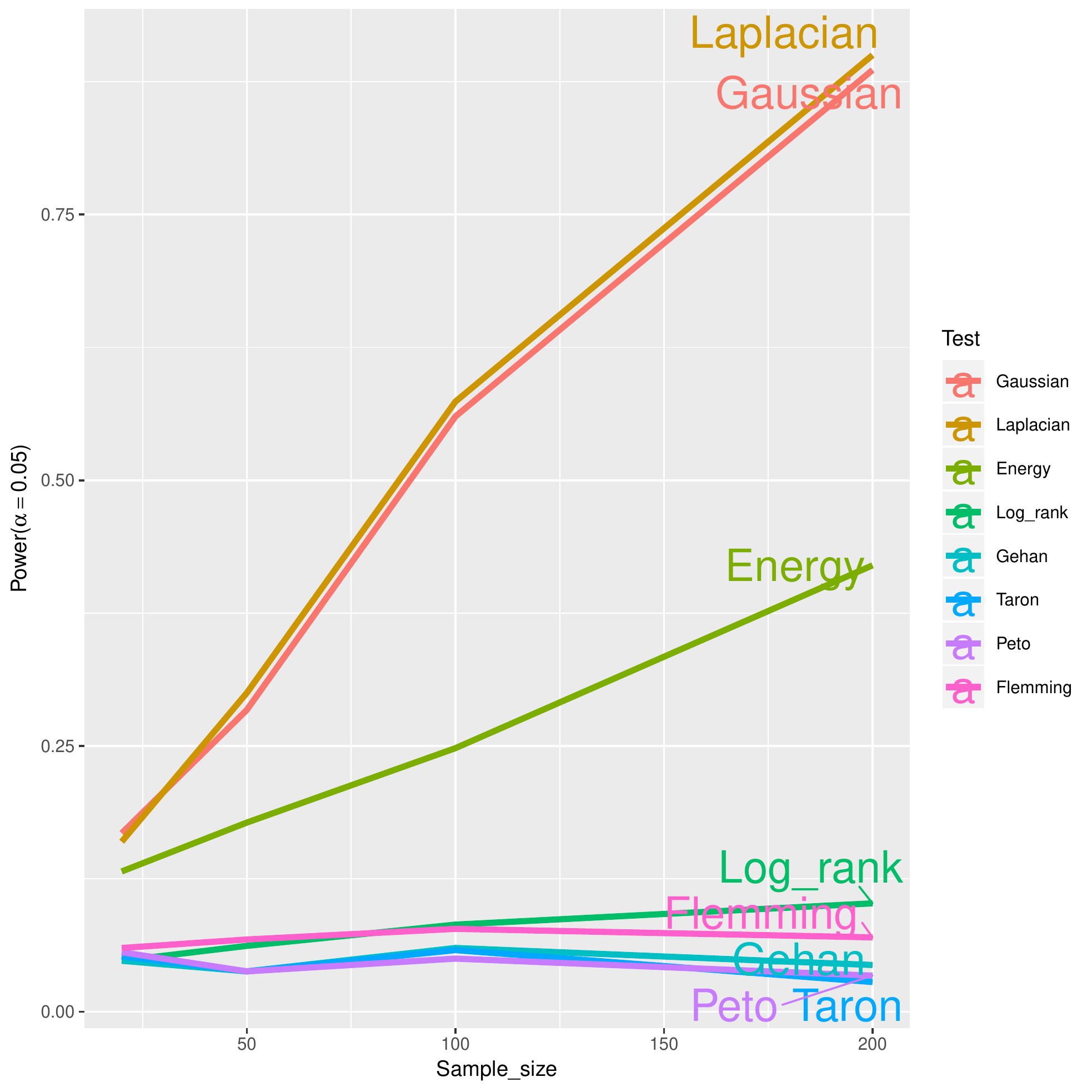}
	\end{tabular}
	\end{gather}
	\caption{Statistical power of survival curves extracted  \cite{motzer2015nivolumab} Figure $2$-B (left) and \cite{ferris2016nivolumab} Figure $1$-B (right).}\label{fig:fig2}
\end{figure}

\begin{figure}[H]
	\centering
    \begin{gather}
	\begin{tabular}{cc}
		\includegraphics[width=8.5cm,height=9cm,keepaspectratio]{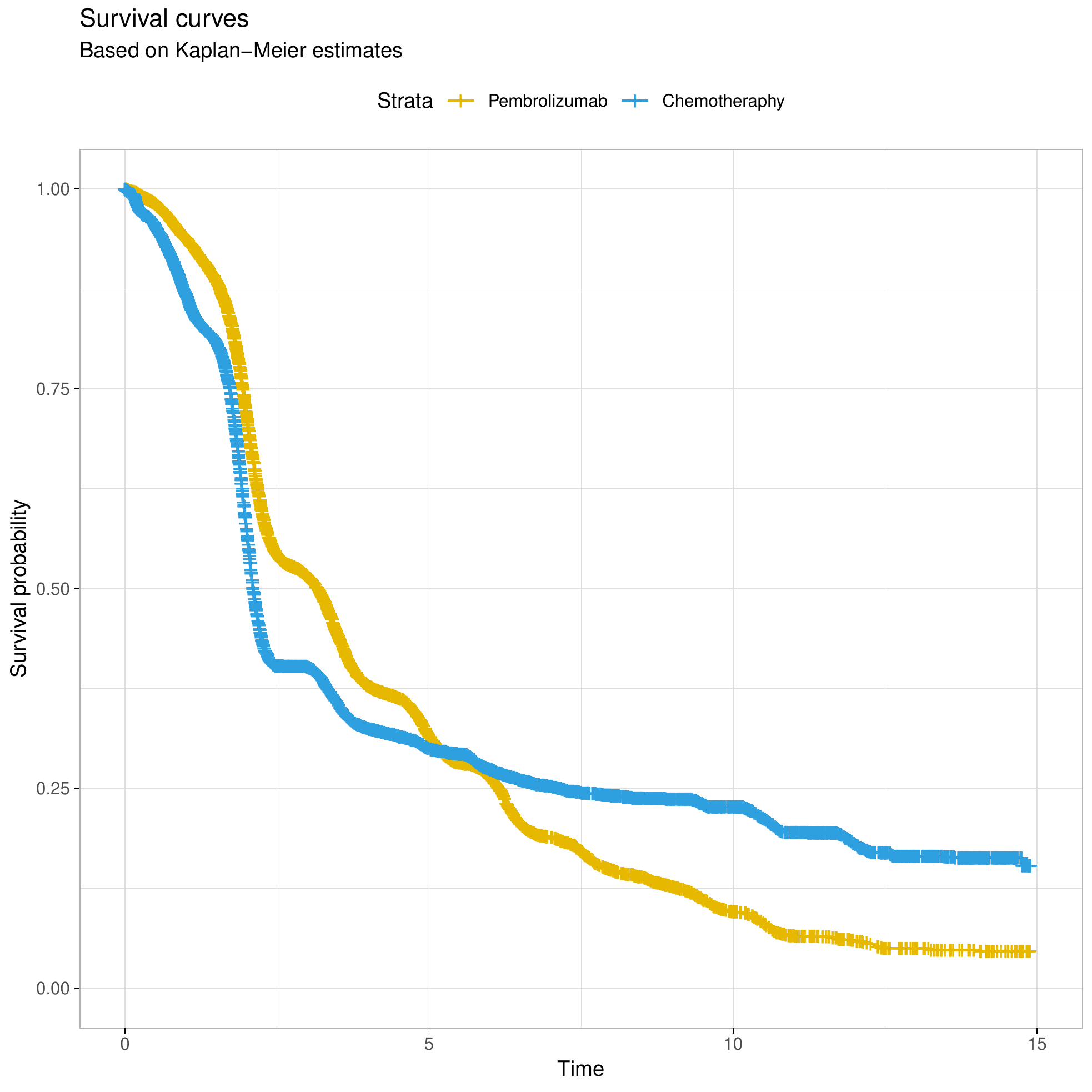}
		\includegraphics[width=8.5cm,height=9cm,keepaspectratio]{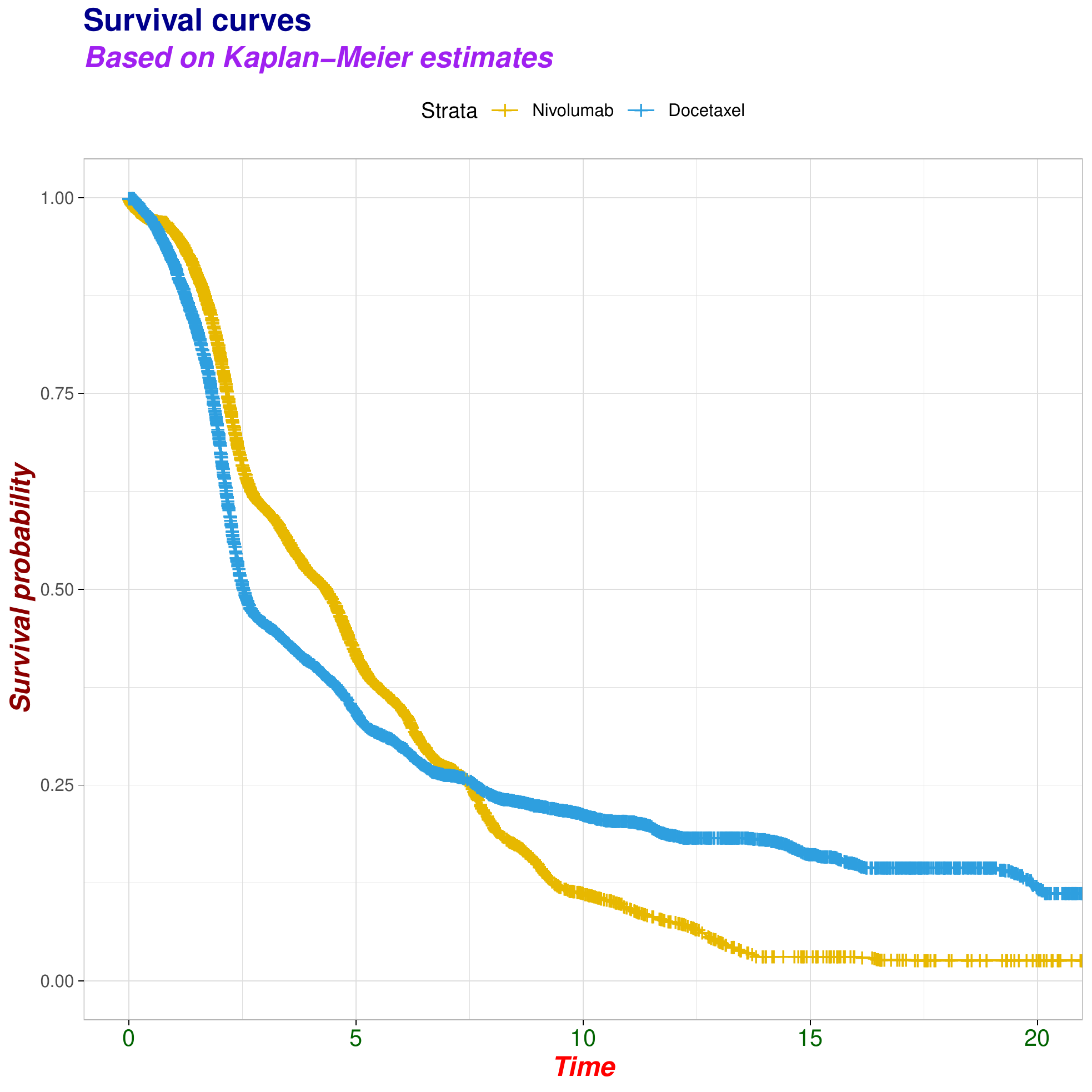}
		\\
		\includegraphics[width=8.5cm,height=9cm,keepaspectratio]{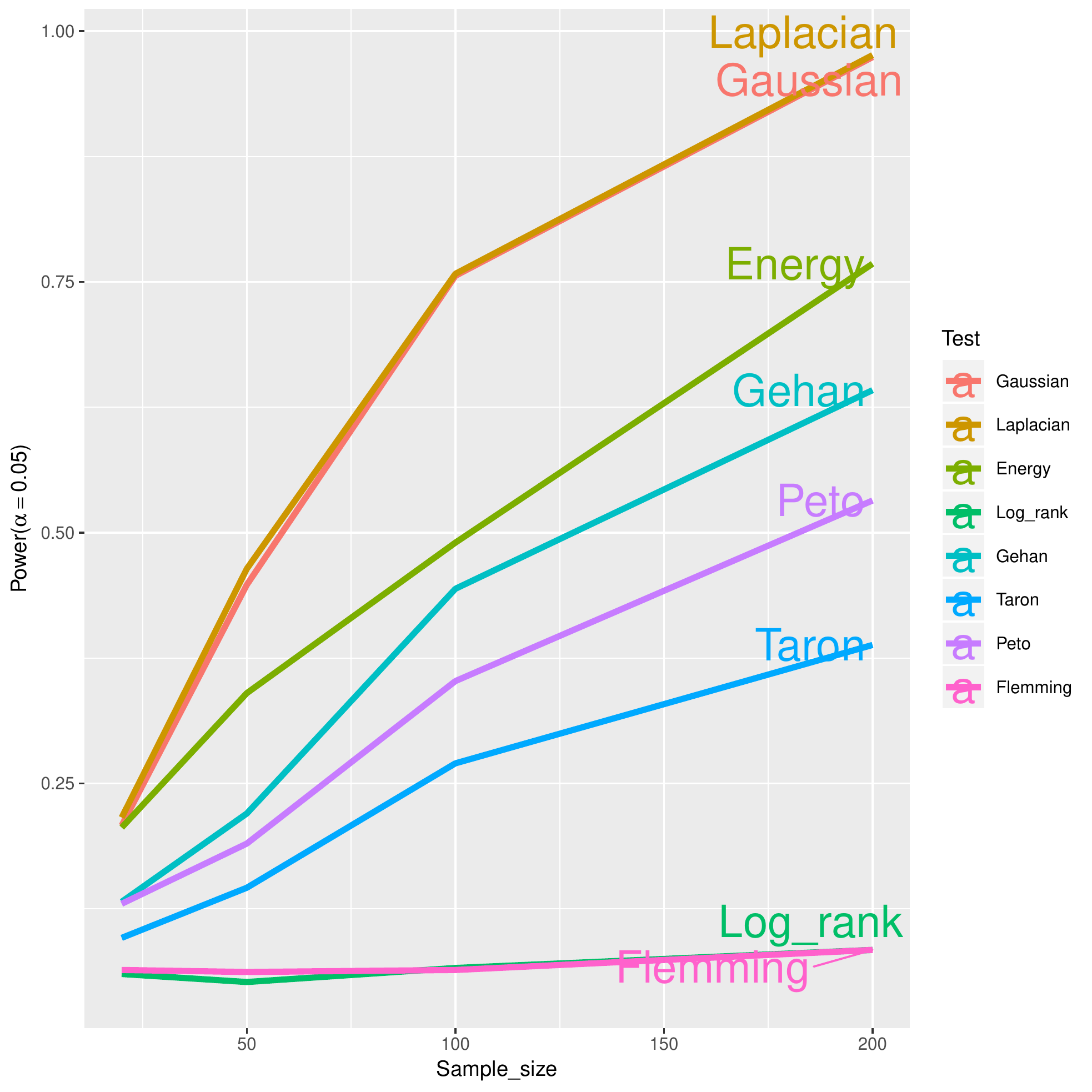}
		\includegraphics[width=8.5cm,height=9cm,keepaspectratio]{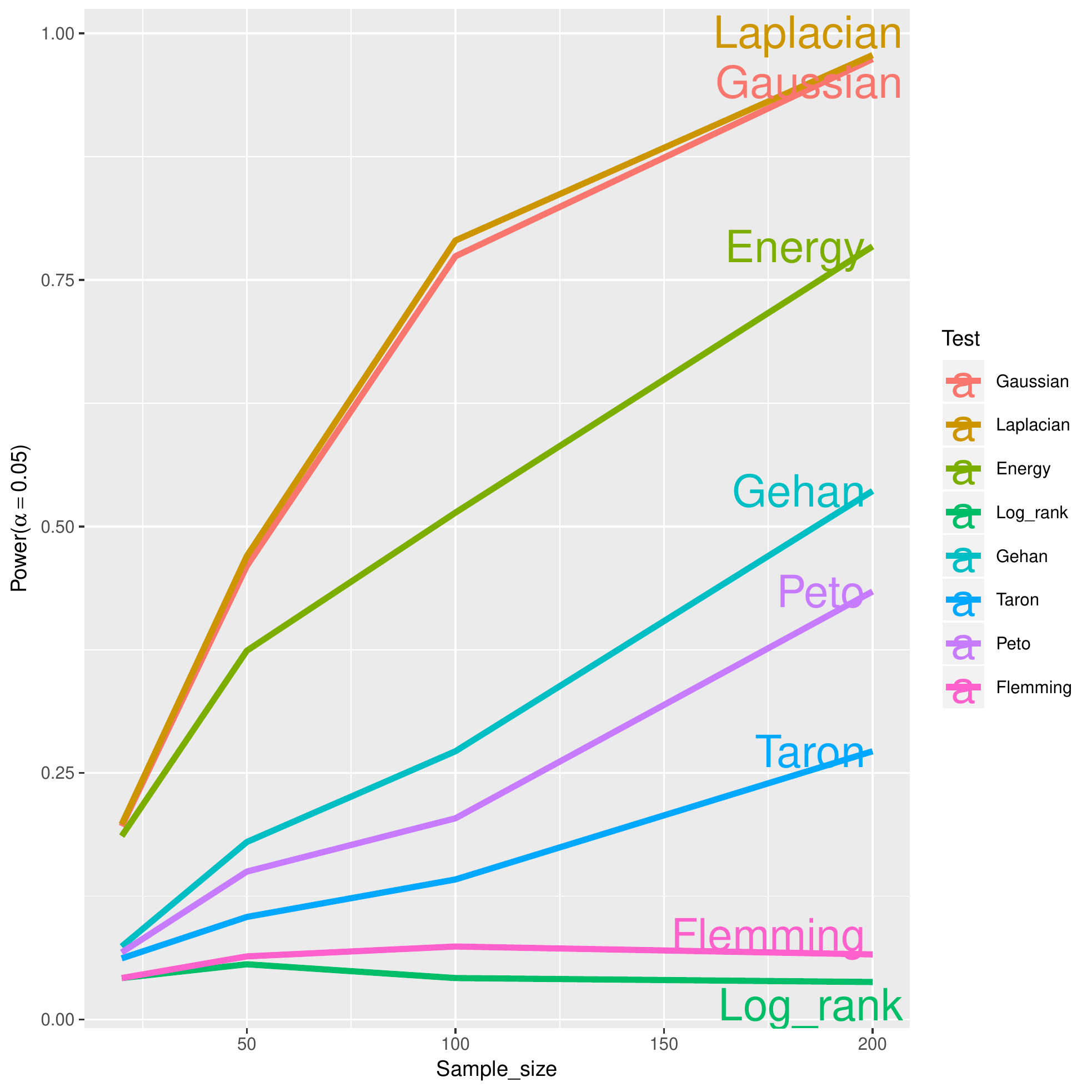}	
	\end{tabular}
	    \end{gather}
	\caption{Statistical power of survival curves extracted \cite{bellmunt2017pembrolizumab} Figure 1-$B$ (left) and \cite{borghaei2015nivolumab} Figure 1-$C$ (right).}\label{fig:fig3}
\end{figure}

\subsubsection{Theoretical proportional hazard ratio in  two population}

We perform $500$ Monte Carlo simulations varying the sample sizes with $50$ individuals from each group, $100$ and $200$, in  the following $11$  cases: $X\sim Exp(1)$ versus $Y\sim Exp(\theta) \text{ }(\text{with} \text{ }\theta\in\{1,1.1,1.2,1.3,1.4,1.5,1.6,1.7,1.8,1.9,2\})$.

We present the results based on the variation of the parameter $\theta$ for each sample size $n$ in  Figure  \ref{fig:fig4}. As we can see in the plots, the log-rank test is usually the most powerful test, as  expected in the situation where this test is optimal from a theoretical point of view. Furthermore, we notice that the use of energy distance with $\alpha=1$ allows us to obtain  better results  in our tests. If we look at Figure \ref{fig:fig1} (right)  where the proportionality hypothesis is not violated either, the results are completely modified. This suggests that the performance of our tests may change completely even if we operate in a context where the hazard functions are proportional.

\begin{figure}[H]
	\centering
	    \hspace*{-11mm}
	\begin{tabular}{cc}
		
		\includegraphics[width=8.5cm,height=8.5cm,keepaspectratio]{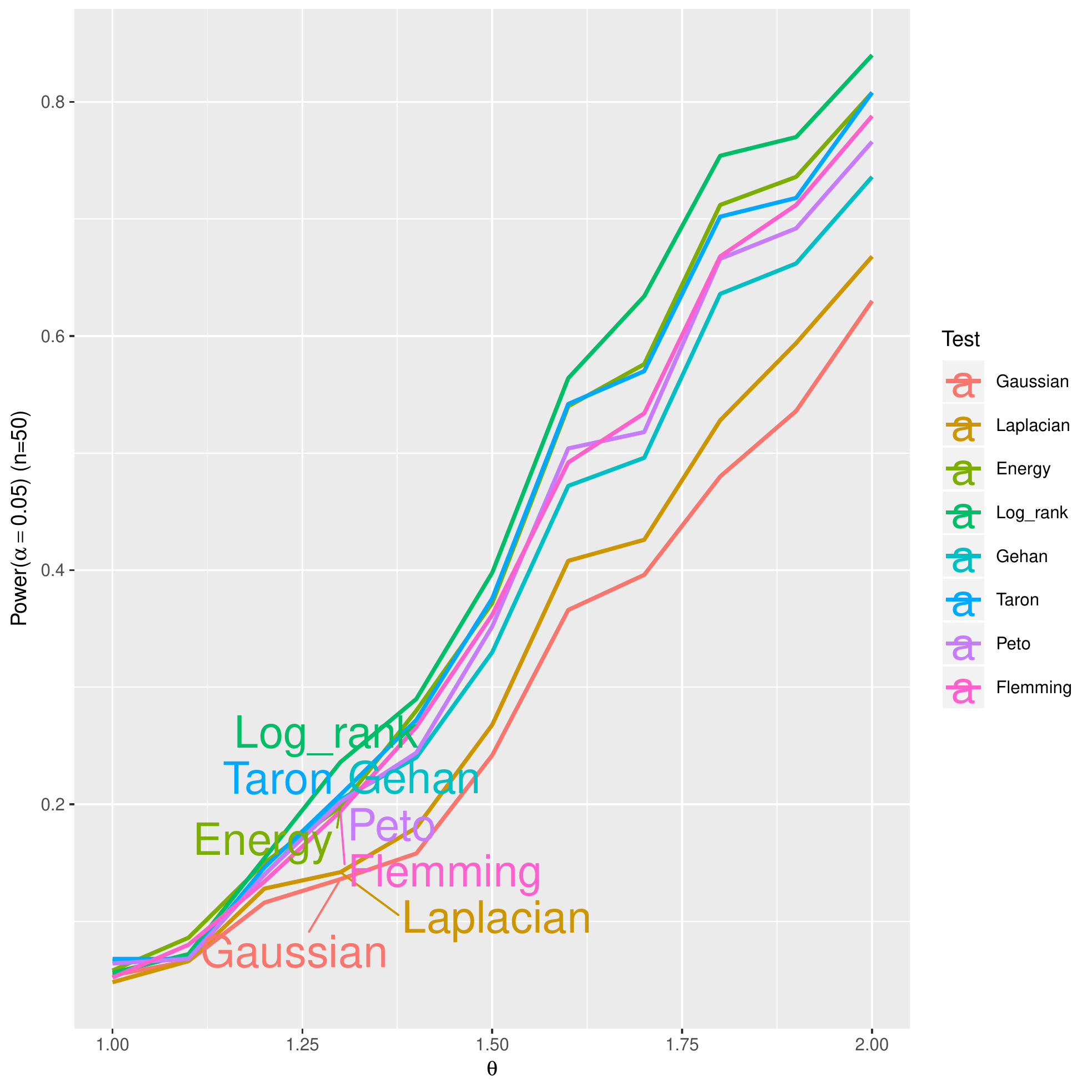}
		\includegraphics[width=8.5cm,height=8.5cm,keepaspectratio]{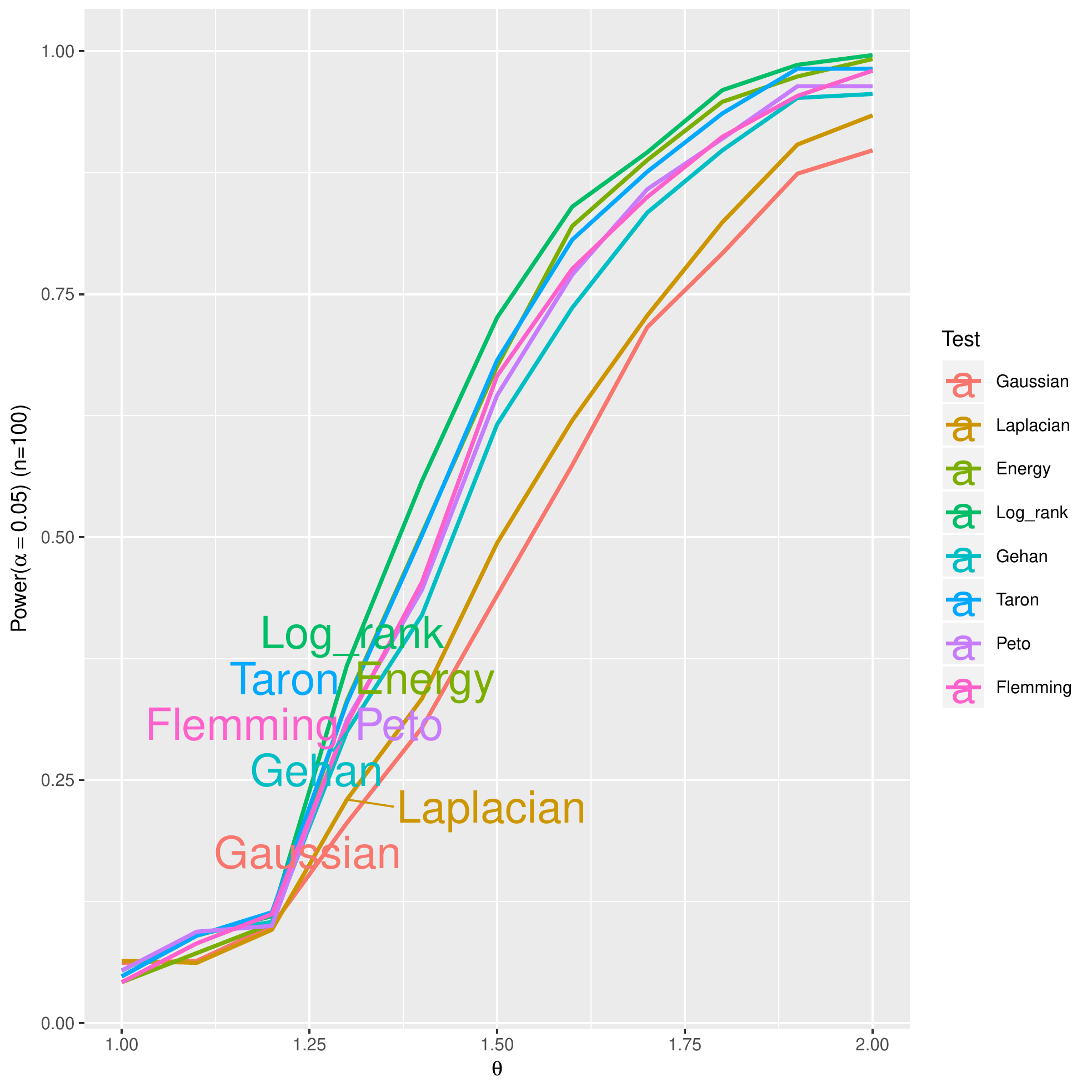}
		
		\\
		\includegraphics[width=9cm,height=9cm,keepaspectratio]{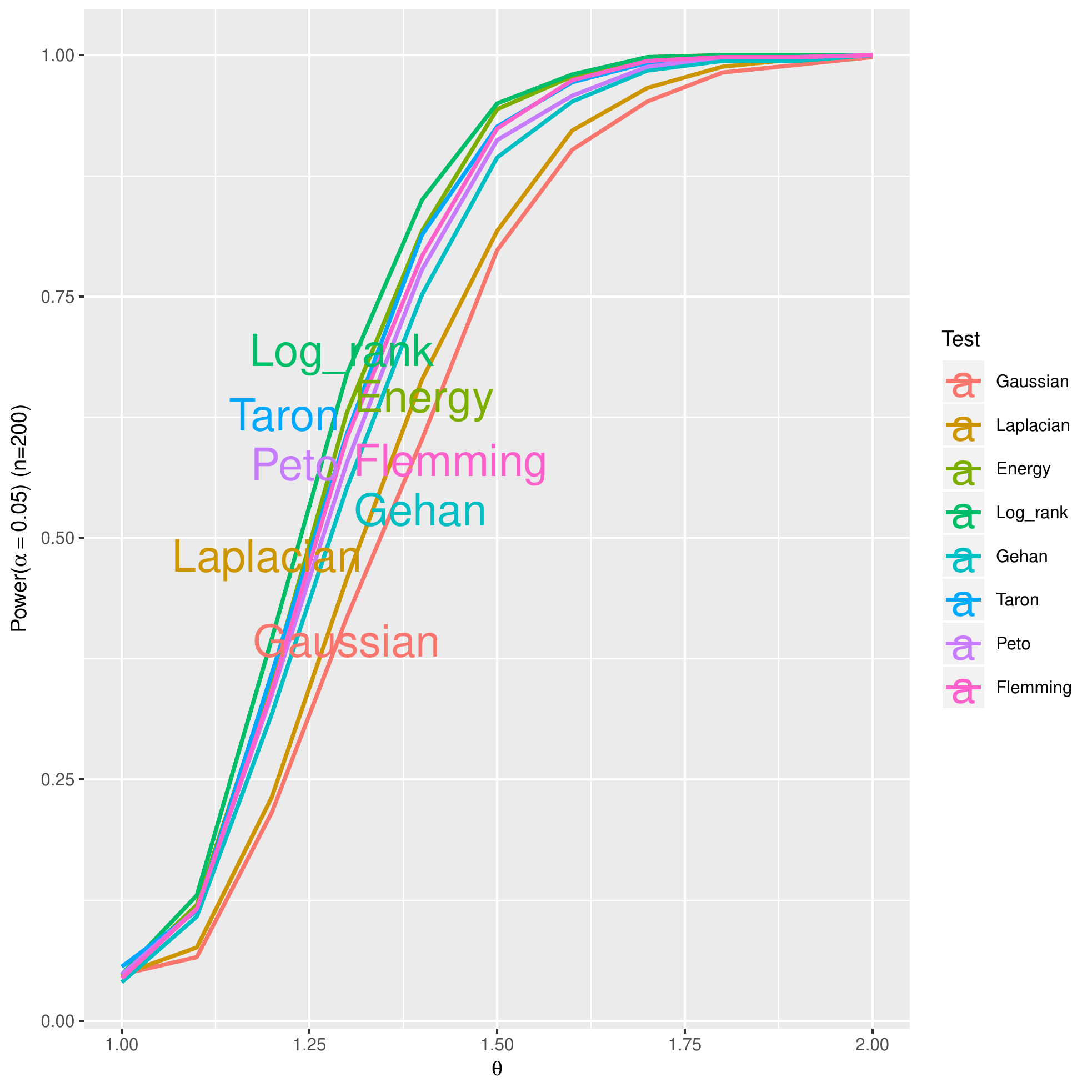}

	\end{tabular}
	
	\caption{Statistical power of study case: Proportional hazard ratio in the two populations.}
	\label{fig:fig4} 
\end{figure}

\section{Example}

To illustrate the potential of the newly proposed  tests in real clinical cases, we use the database from a gastrointestinal tumor study by \cite{stablein1981analysis}. This can be found in the \textbf{\textsf{R}} package ``coin''. The aim of this study is to test whether there are statistically significant differences in the survival curves of two treatments. In Figure \ref{fig:real}, we present the survival curves between the two treatments, observing clear differences between the curves. At first glance, there appears to be a tendency that  the first treatment  increases the long-term survival in comparison to the second. Clearly, the hypothesis of proportional hazards is strongly violated.

\begin{table}[H]
	
	\centering
	\caption{\label{tab:real}$p$-values of the different methods used in the real case}
	\begin{tabular}{rlrlrlr}
		\hline
		& $p$-value &  & $p$-value  & & $p$-value \\
		\hline
		Energy distance $\alpha=1$ & 0.018 & Kernel Gaussian   & 0.004 & Kernel Laplacian  & \textbf{0.002} \\ 
		Logrank & 0.262 &   Gehan & 0.024 & Tarone & 0.075 \\ 
		
		Peto  & 0.030  & Flemming & 0.753 &  \\ 
		\hline
	\end{tabular}
	
\end{table}

\begin{figure}[H]
	\centering
		\scalebox{0.6}{\includegraphics{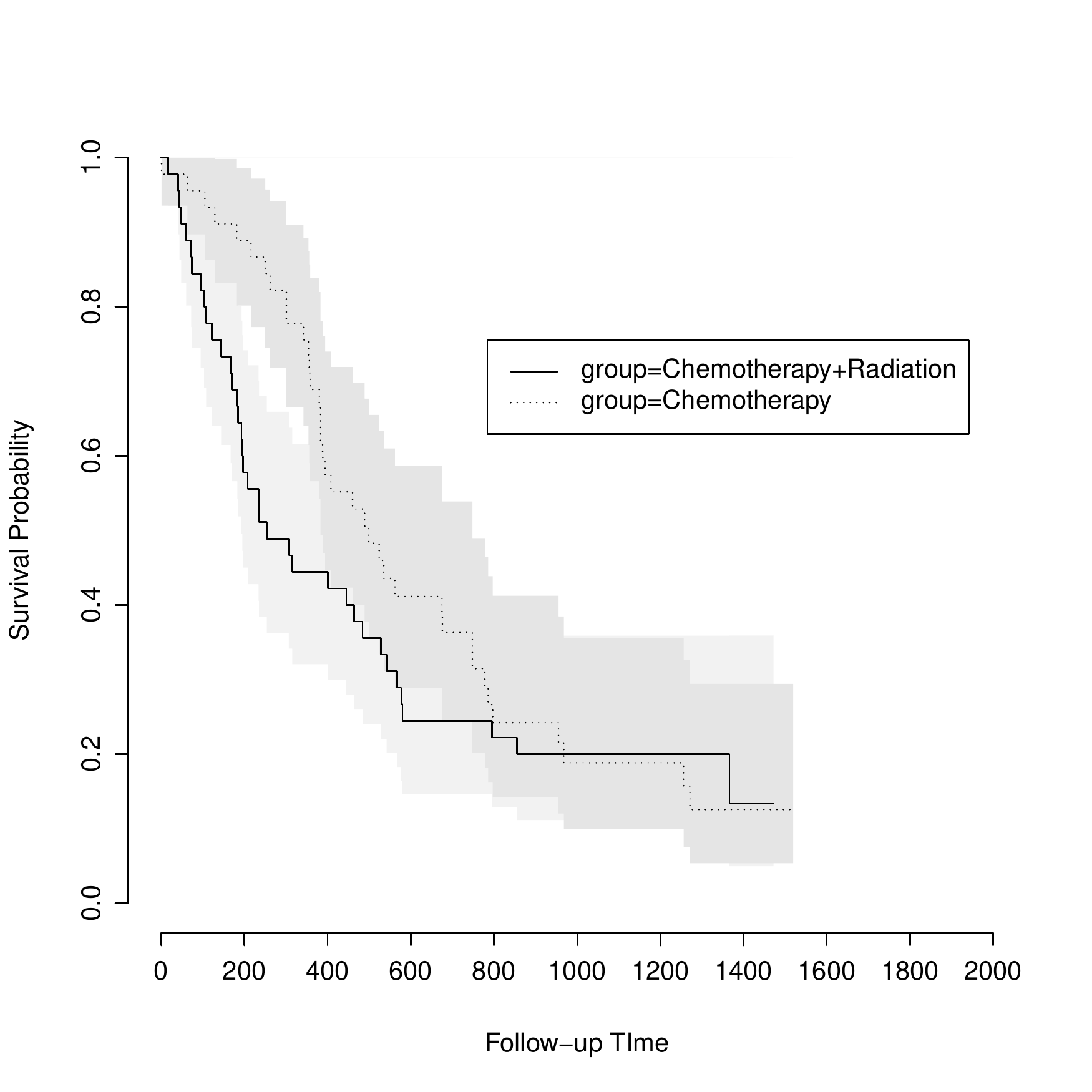}}
	\label{fig:curvasupervivencia}
	\caption{\label{fig:real} Survival curves real case}
\end{figure}

\section{Final remarks}

In this article, a family of consistent tests against all 
alternatives have been proposed to compare the distribution equality between two samples based on the concept of energy distance and kernel mean embeddings. In addition, several theoretical properties of the statistics have been established, together with a set of recommendations on how to select parameters and when to use our tests in situations of clinical interest.

Much work has been done in survival analysis in the context of hazard functions proportionality and for situations where alternatives do not differ much from this situation. In this case, we know that the log-rank tests are optimal \cite{schoenfeld1981asymptotic} and that tests such as Flemming \& Harrington  (\cite{fleming1981class}) offer  a good alternative, choosing a suitable weight function in case of deviations.

If there is evidence that the above situation holds, we do not suggest that  our tests be implemented with the distances/kernels used in this work because the difference in performance with competitors is considerable.

In the  scenario where there is a delay effect on survival in one treatment over another, our tests with the recommended  parameters outperform  the classical tests. These curves share an  important feature: when it is a period of time where one treatment is better than another,  then  this situation is reversed in the long term. In the case where minimal delay effect is observed (as displayed in Figure \ref{fig:fig1} (left)),  the performance of our tests is suboptimal.  However, the situation where our tests have an excellent performance is quite common in clinical trials of immunotherapy \cite{alexander2018hazards} and therefore our tests can be considered an excellent alternative.

The proposed estimators are based on the Kaplan-Meier estimator weights \cite{stute1995statistical}. If there is a high percentage of censored observations along with a small sample size, these methods may not work well (which is very common in all survival analysis methods).  In this case, smoothing the weights may help to increase the power. Alternatively,  if there are apparent  differences at the end of the survival curves, we recommend that one consider  the last observation uncensored \cite{efron1967two}. In either case, this may increase the power of the tests, but also the bias.

The extension of the tests proposed with $k$-samples is analogous to  non-censored methods. There exists a large body of research  in this field, such as Disco analysis \cite{rizzo2010disco} or more recently proposed the kernel methods as in \cite{arm2018kernel}.

The  source code of the new methods is available \url{https://github.com/mmatabuena}. A new  \textbf{\textsf{R}} package  called ``energysurv'' will soon be  launched at \url{https://github.com/mmatabuena} with the proposed methods implemented in \textbf{\textsf{C ++}},  which are believed to be useful for the scientific community. 

\section*{Acknowledgments}

 This work has received financial support from the Conseller\'{i}a de Cultura, Educaci\'{o}n e Ordenaci\'{o}n Universitaria (accreditation 2016-2019, ED431G/08) and the European Regional Development Fund (ERDF). 

The authors are grateful to  Ilia Stepin for his linguistic advice which helped us to improve the clarity and understandability of this article.

\newpage

%
%
%
%
%
%
%
%
%
%
%
%
%
%
%
%
%
%
%
%
%
%
%
%
%

\begin{appendices}
\numberwithin{equation}{section}


\section{Theoretical results}
\label{appendix:proof}

\subsection{Asymptotic  distribution}

The asymptotic  distribution of statistics  under the null hypothesis will be established only for the case of maximum mean discrepancy (MMD). Given the equivalence between the tests based on the kernel mean embeddings and the energy distance \cite{sejdinovic2013equivalence}  this is not restrictive.

%

To begin with, we consider the distribution functions $P_0$, $P_1$, $Q_0$ and $Q_1$ whose meaning was established in the Section \ref{sec:summary} along with the random samples	$\{(X_{j,i},\delta_{j,i})\}_{j=0,1; i=1,\dots,n_j}$.

	Next, let us consider the embeddings $\mu_{P_0}(\cdot)=\int_{0}^{\tau_{0}} K(\cdot,x) dP_{0}(x)\in H_{K}$ and  $\mu_{P_1}(\cdot)=$ $\int_{0}^{\tau_{1}} K(\cdot,x) dP_{1}(x)\in H_{K}$, where $H_{K}$ is the RKHS induced by the kernel $K$ and $\tau_0$, $\tau_1$ are the maximun possible lifetimes defined at the beginning of the Section \ref{tiempos}.

 Under the null hypothesis $P_0=P_1$ and $\tau_0=\tau_1$. Then, $P_0^{\prime}=P_1^{\prime}$ (see  Section \ref{tiempos}), and   $\mu_{P_0^{\prime}}(\cdot)=\mu_{P_1^{\prime}}(\cdot)= \frac{1}{P_0(\tau_0)}\int_0^{\tau_0} K(\cdot,x) dP_0(x)$, where $\mu_{P_0^{\prime}}$ denote the kernel mean embedding \cite{muandet2017kernel} of the distribution $P_0^{\prime}$.

  Given arbitrarily  elements of random sample for each population $X_{0,i},X_{0,j}$ $(i=1,\cdots,n_0,j=1,\cdots,n_0)$,  $X_{1,i^{\prime}},X_{1,j^{\prime}}$ $(i^{\prime}=1,\cdots,n_1,j^{\prime}=1,\cdots,n_1)$,  as $\mu_{P_0^{\prime}}=\mu_{P_1^{\prime}}$, we can replace $K(X_{0,i},X_{0,j})$, $K(X_{1,i^{\prime}},X_{1,j^{\prime}})$ and $K(X_{0,i},X_{1,i^{\prime}})$ with
   $K^{*}(X_{0,i},X_{0,j})$, $K^{*}(X_{1,i^{\prime}},X_{1,j^{\prime}})$ and $K^{*}(X_{0,i},X_{1,i^\prime})$, where  $K^{*}: [0,\tau_0]\times [0,\tau_{0}]\to \mathbb{R}$ is defined as follows

   \begin{gather}
   \begin{align*}
   K^{*}(X_{0,i},X_{0,j}):= <K(X_{0,i},\cdot)-\mu_{P_0^{\prime}},K(X_{0,j},\cdot)-\mu_{P_0^{\prime}}>=  \\
K(X_{0,i},X_{0,j})-\frac{1}{P(\tau_{0})}\int_{0}^{\tau_0}K(X_{0,i},x)dP_0(x)-\frac{1}{P(\tau_{0})}\int_{0}^{\tau_0}K(X_{0,j},x)dP_0(x)+\\ \frac{1}{P_0(\tau_{0})}\frac{1}{P_0(\tau_{0})}   \int_{0}^{\tau_0}K(x,x^{\prime})dP_0(x)dP_0(x^{\prime})=\\
K(X_{0,i},X_{0,j})-E_{X\sim^{} P_0^{\prime}}(K(X_{0,i},X))-E_{X\sim^{} P_0^{\prime}}(K(X_{0,j},X)) +E_{X\sim^{} P_0^{\prime},X^{\prime}\sim^{} P_0^{\prime}}(K(X,X^{\prime})).		
   \end{align*}
      \end{gather}
   
   The previous translation does not change the value between $K(\cdot,\cdot)$ and $K^{*}(\cdot,\cdot)$. This gives the equivalent form of the empirical MMD $\tilde{\gamma}_K^{2}(P_0,P_1)$  (see equation \ref{eqn:estadistico2})
   
  
  \begin{gather}
  \begin{align}
\tilde{\gamma}_K^{2}(P_0,P_1)=  \frac{\sum_{i=1}^{n_0}  \sum_{j\neq i}^{n_0} W^{0}_{i:n_0} W^{0}_{j:n_0} K^{*}(X_{0,(i:n_0)},X_{0,(j:n_0)})}{ \sum_{j=1}^{n_0} \sum_{j\neq i}^{n_0} W^{0}_{i:n_0} W^{0}_{j:n_0}}+ \frac{\sum_{i=1}^{n_1}  \sum_{j\neq i}^{n_1} W^{1}_{i:n_1} W^{1}_{j:n_1}  K^{*}(X_{1,(i:n_1)},X_{1,(j:n_1)})}{\sum_{i=1}^{n_1}  \sum_{j\neq i }^{n_1}W^{1}_{i:n_1} W^{1}_{j:n_1}}\\-2\frac{\sum_{i=1}^{n_0}  \sum_{j=1}^{n_1} W^{0}_{i:n_0} W^{1}_{i:n_1} K^{*}(X_{0,(i:n_0)},X_{1,(j:n_1)})}{\sum_{i=1}^{n_0}  \sum_{j=1}^{n_1} W^{0}_{i:n_0} W^{1}_{i:n_1}}. \nonumber
  \end{align}
  \end{gather}

Now, note  that $K^{*}(\cdot,\cdot)$ is a degenerate kernel

  \begin{gather}
\begin{align*}
E_{X\sim^{} P_0^{\prime}}(K^{*}(X,y))= E_{X}(K(X,y)) - E_{X,X^{\prime}} K(X,X^{\prime}) -E_{X}(K(X,y))+E_{X,X^{\prime}} K(X,X^{\prime})=0 \hspace{0.2cm} \forall y \in [0,\tau_0].
\end{align*}
  \end{gather}

Consequently, in following terms

  \begin{gather}
\begin{align*}
\frac{\sum_{i=1}^{n_0}  \sum_{j\neq i}^{n_0} W^{0}_{i:n_0} W^{0}_{j:n_0} K^{*}(X_{0,(i:n_0)},X_{0,(j:n_0)})}{ \sum_{j=1}^{n_0} \sum_{j\neq i}^{n_0} W^{0}_{i:n_0} W^{0}_{j:n_0}}  \text{ and } \frac{\sum_{i=1}^{n_1}  \sum_{j\neq i}^{n_1} W^{1}_{i:n_1} W^{1}_{j:n_1}  K^{*}(X_{1,(i:n_1)},X_{1,(j:n_1)})}{\sum_{i=1}^{n_1}  \sum_{j\neq i }^{n_1}W^{1}_{i:n_1} W^{1}_{j:n_1}} \nonumber,
\end{align*}
  \end{gather}

we can apply the limits theorems for U-statistics under right censored data \cite{BOSE200284,fernndez2018kaplanmeier}. In particular, we will use the results \cite{fernndez2018kaplanmeier} under the weakest conditions to use the theorems. Under the conditions  assumed in Section \ref{sec:summary} along with the Euclidean distance and kernel of Table \ref{tab:tabla1},  we   can apply  the theoretical results directly.

By the Corollary $2.9$ \cite{fernndez2018kaplanmeier}, under the null hyphotesis and $\tau_0= \tau_1$,  we have:

$$\frac{\sum_{i=1}^{n_0}  \sum_{j\neq i}^{n_0} W^{0}_{i:n_0} W^{0}_{j:n_0} K^{*}(X_{0,(i:n_0)},X_{0,(j:n_0)})}{ \sum_{j=1}^{n_0} \sum_{j\neq i}^{n_0} W^{0}_{i:n_0} W^{0}_{j:n_0}}  \overset{D}{\to} c_1 + \psi$$
and 
$$\frac{\sum_{i=1}^{n_1}  \sum_{j\neq i}^{n_1} W^{1}_{i:n_1} W^{1}_{j:n_1}  K^{*}(X_{1,(i:n_1)},X_{1,(j:n_1)})}{\sum_{i=1}^{n_1}  \sum_{j\neq i }^{n_1}W^{1}_{i:n_1} W^{1}_{j:n_1}}  \overset{D}{\to} c_2 + \psi$$
where $\psi= \sum_{i=1}^{\infty} \lambda_i (\epsilon_i^{2}-1)$, with $\epsilon_{i}$ $i.i.d$ standard normal random variables and $c_1$, $c_2$ are two constants specified in \cite{fernndez2018kaplanmeier} that are not relevant for our purposes.

The structure of the previous limits coincides with the case without censoring  in the degenerate case. More concretely, the limit is $c+\psi$ \cite{korolyuk2013theory} where $c$ is a constant.

 However, for the term

$$2\frac{\sum_{i=1}^{n_0}  \sum_{j=1}^{n_1} W^{0}_{i:n_0} W^{1}_{i:n_1} K^{*}(X_{0,(i:n_0)},X_{1,(j:n_1)})}{\sum_{i=1}^{n_0}  \sum_{j=1}^{n_1} W^{0}_{i:n_0} W^{1}_{i:n_1}}$$

we have a U-statistics of two samples under right censored data in the degenerate case. There are no  theoretical results in the literature.





The deduction of the limits theory in this case is beyond the scope of this work, and will be presented in another paper. In any case, the limiting distribution coincides with the case without censorship. This is

$$ \sqrt{n_0n_1}  \frac{\sum_{i=1}^{n_0}  \sum_{j=1}^{n_1} W^{0}_{i:n_0} W^{1}_{i:n_1} K^{*}(X_{0,(i:n_0)},X_{1,(j:n_1)})}{\sum_{i=1}^{n_0}  \sum_{j=1}^{n_1} W^{0}_{i:n_0} W^{1}_{i:n_1}} \overset{D}{\to} \eta_{\infty}$$

$$\eta_{\infty}=\sum_{j=1}^{\infty} \lambda_j \tau_j \epsilon_j,$$
where $\{\tau_j\}^{\infty}_{j=1}$ and $\{\epsilon_j\}^{\infty}_{j=1}$ are two independence sequences of standard normal random variables.

\subsection{Consistency against all alternatives}

\begin{theorem}\label{positivo}
	Let $S,A$ be  arbitrary metrics spaces with the same topology defined on $\mathbb{R^{+}}$, $S$ contained on $A$ and let $\gamma(x,y)$ be a continuous, symmetric, real function on  $A\times A$. Suppose $X$,$X^{\prime}$, $Y$,$Y^{\prime}$ are independent random variables, $X$,$X^{\prime}$ identically distributed, and $Y$,$Y^{\prime}$ are identically distributed. We suppose, moreover that,
	$\gamma(X,X^{\prime})$, $\gamma(Y,Y^{\prime})$, and $\gamma(X,Y)$ have finite expected values on $A$. Then
	
	\begin{gather}
	\begin{align*}
	2\frac{\int_S \int_S \gamma(x,y)dP(x) dQ(y)}{\int_{S}dP(x)\int_{S}dQ(y)}-	\frac{\int_S \int_S \gamma(x,y)dP(x) dP(y)}{(\int_{S}dP(x))^{2}}-\frac{\int_S \int_S \gamma(x,y)dQ(x) dQ(y)}{(\int_{S}dQ(x))^{2}}\geq 0
	\end{align*}
		\end{gather}
	
	if and only if $\phi$ is negative-definite and where $P$ and $Q$ denote the distribution of $X$ and $Y$ respectively. If $\gamma$ is strictly negative then equality holds if and only if $X$ and $Y$ are identically distributed on $S$.   
\end{theorem}

\begin{proof}
	%
	%
	
	By Theorem $1$ \cite{szekely2005new}, it is verified:

	\begin{equation}
	2\int_A \int_A \gamma(x,y)dP(x) dQ(y)-	\int_A	 \int_A \gamma(x,y)dP(x) dP(y)-\int_A \int_A \gamma(x,y)dQ(x) dQ(y)\geq 0,
	\end{equation}	
	if and only if $\phi$ is negative-definite. If $\gamma$ is strictly negative then equality holds if and only if $X$ and $Y$ are identically distributed on $A$.
	
	Now, we define the following random variables on $S$, $X^{*}$, $Y^{*}$ with distribution function $P^{\prime}$, $Q^{\prime}$ respectively, as follows :
	
	$dP^{\prime}(x)= c_1 dP(x)$ and $dQ^{\prime}(x)= c_2 dP(x)$, where $c_1= \frac{1}{\int_{S} dP(x)}$ and $c_2= {\frac{1}{\int_{S} dQ(x)}}$, 
	and we consider their copies $X^{*\prime}$,$Y^{*\prime}$.  As $\gamma(X,X^{\prime})$, $\gamma(Y,Y^{\prime})$, and $\gamma(X,Y)$ have finite expected values in $A$, then $\gamma(X^{*},X^{*\prime})$, $\gamma(Y^{*},Y^{*\prime})$, and $\gamma(X^{*},Y^{*})$ have finite expected values in $S$. Moreover, let $\gamma(x,y)$ be a continuous, symmetric, real function in  $S\times S$.

	This leads to:

	\begin{gather}
	\begin{align*}
	2c_1c_2\int_S \int_S \gamma(x,y)dP(x) dQ(y)-	c_1^{2}\int_S \int_S \gamma(x,y)dP(x) dP(y)-c_2^{2}\int_S \int_S \gamma(x,y)dQ(x) dQ(y)\geq 0.
	\end{align*}
		\end{gather}

	if and only if $\phi$ is negative definite,  and 
	
		\begin{gather}	
	\begin{align*}
	2c_1c_2\int_S \int_S \gamma(x,y)dP(x) dQ(y)-	c_1^{2}\int_S \int_S \gamma(x,y)dP(x) dP(y)-c_2^{2}\int_S \int_S \gamma(x,y)dQ(x) dQ(y)= 0.
	\end{align*}
			\end{gather}
	
	If $X^{*}$ and $Y^{*}$ are identically distributed in $S$ (with $\phi$ being  strictly negative) or equivalent $P(t)= Q(t)$ $\forall t$  $\in S$.

\end{proof}

\begin{theorem}{\label{consistencia}}
	
	Let $X_{j,i}= min(T_{j,i},C_{j,i})\sim^{i.i.d. } P_{c(j)} $ and  $\delta_{j,i}= 1\{X_{j,i}= T_{j,i}\}$	$(j=0,1; i=1,\dots,n_j)$ with $P_{c(j)}$ $(j=0,1)$ and under the conditions  assumed in Section \ref{sec:summary} imposed on the variables $T_{j,i}\sim^{i.i.d.} P_{j},C_{j,i}\sim^{i.i.d.} Q_{j}$ $(j=0,1; i=1,\dots,n_j)$. Then:

	\begin{gather}
	\begin{align*}
	\tilde{\epsilon}_\alpha(P_0,P_1) \overset{n_0,n_1\to \infty}{\to} \epsilon_{c(\alpha)}(P_0,P_1)= 2\frac{\int_{0}^{\tau_0} \int_{0}^{\tau_1}	 ||x-y||^{\alpha} dP_0^{*}(x) dP_1^{*}(y)}{\int_{0}^{\tau_0} \int_{0}^{\tau_1} dP_0^{*}(x) dP_1^{*}(y)}\\
	-\frac{\int_{0}^{\tau_0} \int_{0}^{\tau_0}	 ||x-y||^{\alpha} dP_0^{*}(x) dP_0^{*}(y)}{\int_{0}^{\tau_0} \int_{0}^{\tau_0} dP_0^{*}(x) dP_0^{*}(y)} -\frac{\int_{0}^{\tau_1}\int_{0}^{\tau_1}	 ||x-y||^{\alpha} dP_1^{*}(x) dP_1^{*}(y)}{\int_{0}^{\tau_1}\int_{0}^{\tau_1} dP_1^{*}(x) dP_1^{*}(y)},
	\end{align*} 
		\end{gather}
		\begin{gather}
	\begin{align*}
	\tilde{\gamma}_K(P_0,P_1) \overset{n_0,n_1\to \infty}{\to} \gamma_{c(K)}(P_0,P_1)= 2  \frac{\int_{0}^{\tau_0} \int_{0}^{\tau_1}	 K(x,y) dP_0^{*}(x) dP_1^{*}(y)}{\int_{0}^{\tau_0} \int_{0}^{\tau_1} dP_0^{*}(x) dP_1^{*}(y)}\\
	-\frac{\int_{0}^{\tau_0} \int_{0}^{\tau_0}	 K(x,y) dP_0^{*}(x) dP_0^{*}(y)}{\int_{0}^{\tau_0} \int_{0}^{\tau_0}	  dP_0^{*}(x) dP_0^{*}(y)}-\frac{\int_{0}^{\tau_1}\int_{0}^{\tau_1}	K(x,y) dP_1^{*}(x) dP_1^{*}(y)}{\int_{0}^{\tau_1}\int_{0}^{\tau_1} dP_1^{*}(x) dP_1^{*}(y)},
	\end{align*} 
			\end{gather}

	where
	
	$P_0^{*}(x)= \left\{ \begin{array}{lcc}
	P_0(x) &   if  & x < \tau_{0} \\
	\\ P_0(\tau_0^{-})+1\{\tau_0\in A^{1}\}P_0(\tau_0) &  if & x \geq \tau_0 \\
	\end{array}
	\right.$
	
	and
	
	$P_1^{*}(x)= \left\{ \begin{array}{lcc}
	P_1(x) &   if  & x < \tau_{1} \\
	\\ P_1(\tau_1^{-})+1\{\tau_1\in A^{1}\}P_1(\tau_1) &  if & x \geq \tau_1. \\
	
	\end{array}
	\right.$
	
	Here,
	$\tau_0=\inf\{x:1-P_{c(0)}(x)=0\}$, $\tau_1=\inf\{x:1-P_{c(1)}(x)=0\}$, $A^{0}=\{x\in \mathbb{R}| P_{c(0)}\{x\}>0\}$ and $A^{1}=\{x\in \mathbb{R}| P_{c(1)}\{x\}>0\}$. 
\end{theorem}
\begin{proof}
	
	%

	The proof consists of repeatedly applying the strong laws of large numbers for $U$ Kaplan-Meier statistics with two samples \cite{stute1993multi}, with the convergence results for the $U$ statistic of degree two for randomly censored data \cite{bose1999strong}. 
	
	According to \cite{stute1993multi}:

	$$ \sum_{i=1}^{n_1} \sum_{j=1}^{n_1} W^{0}_{i:n_0} W^{1}_{i:n_1} h(X_{0,(i:n_0)},X_{1,(j:n_1)}) \overset{n_0,n_1\to \infty}{\to} \int_{0}^{\tau_0} \int_{0}^{\tau_1} h(x,y) dP_0^{*}(x)dP_1^{*}(y) $$
	where $h$ is a given kernel of degree two such that	
	$$ \int h(x,y) dP_0(x) dP_1(y)< \infty. $$

	Note, by hypothesis,  $ P_{c(j)}$ $(j = 0,1) $ is a continuous distribution function. Then, $A^{0}$ and $A^{1}$ are empty sets, and therefore $P^{*}_0(x)= P_0(x)$ $\forall \in[0,\tau_0]$  and $P^{*}_1(x)= P_1(x)$ $\forall\in[0,\tau_1]$.

	Applying the previous result with $h(x,y)=1$ to the following expressions, along with the properties of convergence in probability,  we have:
	
	$$    \frac{\sum_{i=1}^{n_1} \sum_{j=1}^{n_1} W^{0}_{i:n_0} W^{1}_{i:n_1} h(X_{0,(i:n_0)},X_{1,(j:n_1)})}{\sum_{i=1}^{n_1} \sum_{j=1}^{n_1} W^{0}_{i:n_0} W^{1}_{i:n_1}} \overset{n_0,n_1\to \infty}{\to} \frac{\int_{0}^{\tau_0} \int_{0}^{\tau_1} h(x,y) dP_0^{*}(x) dP_1^{*}(y)}{\int_{0}^{\tau_0} \int_{0}^{\tau_1}  dP_0^{*}(x) dP_1^{*}(y)}. $$

	Using Theorem $1$ of \cite{bose1999strong},  it is also verified that

	$$\frac{\sum_{i=1}^{n_0}  \sum_{j\neq i}^{n_0} W^{0}_{i:n_0} W^{0}_{j:n_0} h(X_{0,(i:n_0)},X_{0,(j:n_0)})}{\sum_{i=1}^{n_0}  \sum_{j\neq i}^{n_0} W^{0}_{i:n_0} W^{0}_{j:n_0}} \overset{n_0 \to \infty}{\to} \frac{\int_{0}^{\tau_0} \int_{0}^{\tau_0} h(x,y)	  dP_0^{*}(x) dP_0^{*}(y)}{\int_{0}^{\tau_0} \int_{0}^{\tau_0} dP_0^{*}(x) dP_0^{*}(y)}, $$
	     and	
	$$\frac{\sum_{i=1}^{n_1}  \sum_{j\neq i}^{n_1} W^{1}_{i:n_1} W^{1}_{j:n_1} h(X_{1,(i:n_1)},X_{1,(j:n_1)})}{\sum_{i=1}^{n_1}  \sum_{j\neq i}^{n_1} W^{1}_{i:n_1} W^{1}_{j:n_1}} \overset{n_1 \to \infty}{\to} \frac{\int_{0}^{\tau_1} \int_{0}^{\tau_1} h(x,y)	  dP_1^{*}(x) dP_1^{*}(y)}{\int_{0}^{\tau_1} \int_{0}^{\tau_1} dP_1^{*}(x) dP_1^{*}(y)}. $$
	
	Finally, taking  $h(x,y)$ as  $||x-y||^{\alpha}$ or $h(x,y)=K(x,y)$ and applying the properties of convergence  in probability of the sum of two random variables, the desired result is obtained.
	
\end{proof}

\begin{theorem}
		Let  $X_{j,i}= min(T_{j,i},C_{j,i})\sim^{ \text{i.i.d.}  } P_{c(j)} $ and  $\delta_{j,i}= 1\{X_{j,i}= T_{j,i}\}$	$(j=0,1; i=1,\dots,n_j)$ with $P_{c(j)}$ $(j=0,1)$. Suppose also that  the conditions  stated in Section \ref{sec:summary} hold  for the random variables  $T_{j,i}\sim^{ \text{i.i.d.}  }  P_{j},\,\,\,\,\,C_{j,i} \sim^{ \text{i.i.d.}  }  Q_{j}$ $(j=0,1; i=1,\dots,n_j)$ . Further assume that $\tau_0=\tau_1$  or the  support of the distribution functions $P_0$ and $P_1$ is  contained in the intervals $ [0,\tau_0] $ and $[0,\tau_1]$, respectively. Then, 
	for testing the null  $H_0  \,:\, P_{0}(t)  =  P_{1}(t) \,\,\,\,\,\forall  t\in [0,\tau_1] $	
	the statistics $T_{\tilde{\epsilon}_\alpha}$ and   $T_{\tilde{\gamma}_K^{2}}$  determine tests that are consistent against all fixed  alternatives with continuous random variables.

\end{theorem}

\begin{proof}

	We assume it without any restrictions that $P_0$ and $ P_1$ have the same support (otherwise it is enough to extend the probability measure with less support to the higher one). If $\tau_0= \tau_1$, we can  apply Theorems \ref{positivo}-\ref{consistencia} and then we have it guaranteed that:
	
	\begin{gather}\label{cota1}
	\begin{align}
	\lim_{n_0\to \infty,n_1 \to \infty} \tilde{\epsilon}_\alpha(P_0,P_1)= 
	2\frac{\int_{0}^{\tau_0} \int_{0}^{\tau_1}	 ||x-y||^{\alpha} dP_0^{*}(x) dP_1^{*}(y)}{\int_{0}^{\tau_0} \int_{0}^{\tau_1} dP_0^{*}(x) dP_1^{*}(y)}\\
	-\frac{\int_{0}^{\tau_0} \int_{0}^{\tau_0}	 ||x-y||^{\alpha} dP_0^{*}(x) dP_0^{*}(y)}{\int_{0}^{\tau_0} \int_{0}^{\tau_0} dP_0^{*}(x) dP_0^{*}(y)} -\frac{\int_{0}^{\tau_1}\int_{0}^{\tau_1}	 ||x-y||^{\alpha} dP_1^{*}(x) dP_1^{*}(y)}{\int_{0}^{\tau_1}\int_{0}^{\tau_1} dP_1^{*}(x) dP_1^{*}(y)} \geq 0 \nonumber
	\end{align}
	\end{gather}
	%
	%
	
	\begin{align}\label{cota2}
	\lim_{n_0\to \infty,n_1 \to \infty} \tilde{\gamma}_K(P_0,P_1)= \frac{\int_{0}^{\tau_0} \int_{0}^{\tau_0}	 K(x,y) dP_0^{*}(x) dP_0^{*}(y)}{\int_{0}^{\tau_0} \int_{0}^{\tau_0}	  dP_0^{*}(x) dP_0^{*}(y)} \\
	+\frac{\int_{0}^{\tau_1}\int_{0}^{\tau_1}	K(x,y) dP_1^{*}(x) dP_1^{*}(y)}{\int_{0}^{\tau_1}\int_{0}^{\tau_1} dP_1^{*}(x) dP_1^{*}(y)}  -2   \frac{\int_{0}^{\tau_0} \int_{0}^{\tau_1}	 K(x,y) dP_0^{*}(x) dP_1^{*}(y)}{\int_{0}^{\tau_0} \int_{0}^{\tau_1} dP_0^{*}(x) dP_1^{*}(y)} \geq 0  \nonumber.
	\end{align}

		Furthermore \ref{cota1} and \ref{cota2} are equal to zero if and only if  $P_0(t)= P_1(t)$ $\forall t \in [0,\tau_1]$.

	Suppose $\exists t \in [0,\tau_1]$ $P_0 (t) \neq P_1 (t)$, then we have strict	inequality in \ref{cota1} and \ref{cota2}, so  with probability one  $\lim_{n_0\to \infty,n_1 \to \infty} P(\tilde{\epsilon}_{\alpha}(P_0,P_1) = c_{\epsilon_\alpha}>0)=1$ and $\lim_{n_0\to \infty,n_1 \to \infty} P(\tilde{\gamma}_K(P_0,P_1)= c_K>0)=1$.  
	According to the theory of degenerate $U$-statistics \cite{korolyuk2013theory}  under the null hyphotesis, there exist constants $c_{\alpha_1}$ and  $c_{\alpha_2}$  satisfying

	\begin{equation*}
	\lim_{n\to \infty} P(\frac{n_0 n_1}{n_{0}+n_{1}} \hat{\epsilon}_{\alpha}(P_0,P_1) > c_{\alpha_1})= \alpha  \text{ and } \lim_{n\to \infty} P(\frac{n_0n_1}{n_{0}+n_1} \hat{\gamma}_K(P_0,P_1) > c_{\alpha_2})=\alpha.
	\end{equation*} 
	
	Under the alternative hypothesis
	
	\begin{equation*}
	\lim_{n\to \infty} P(\frac{n_0 n_1}{n_{0}+n_{1}} \hat{\epsilon}_{\alpha}(P_0,P_1) > c_{\alpha_1})=1  \text{ and } \lim_{n\to \infty} P(\frac{n_0n_1}{n_{0}+n_1} \hat{\gamma}_K(P_0,P_1) > c_{\alpha_2})=1
	\end{equation*} 
	since $n\hat{\epsilon}_{\alpha}(P_0,P_1) \to \infty$ and $ n\hat{\gamma}_K(P_0,P_1)$ with probabiliy one as $n\to \infty$.

	%

	In the case of $\tau_0 \neq \tau_1$ the  support of the distribution functions $P_0$ and $P_1$ is  contained in the intervals $ [0,\tau_0] $ and $[0,\tau_1]$, and, in this situation, the normalization constants are $1$, and then, the previous argument is  true.
\end{proof}

\section{V-statistics as a distance between samples}
\label{appendix:B}
We will now establish that the statistics defined in (\ref{eqn:estadistico3})--(\ref{eqn:estadistico4})  behave like distances between the elements of the sample $\{(X_{j,i},\delta_{j,i})\}_{j=0,1; i=1,\dots,n_j}$ defined in Section \ref{sec:summary}.

Given two arbitrary samples  $A:=\{(X_{j,i},\delta_{j,i})\}_{j=0; i=1,\dots,n_0}$ and $B:=\{(X_{j,i},\delta_{j,i})\}_{j=1; i=1,\dots,n_1}$, a function $d: (\mathbb{R}^{+}\times \{0,1\})^{n_0}\times (\mathbb{R}^{+}\times \{0,1\})^{n_1} \to \mathbb{R}$ between  $A$ and $B$ is a distance if:

\begin{itemize}
	\item $d(A,B)\geq 0$ and $d(A,B)=0$ iff $A=B$. 
	\item $d(A.B)= d(B,A)$.
	\end{itemize}

Moreover, given an arbitrary sample $C$, it is verified that:

\begin{itemize}
	\item $d(A,B)\leq d(A,C)+d(B,C)$.
\end{itemize}

The population version of energy distance and maximum mean discrepancy with appropiate distances/kernel (for example, with euclidean distance) verify those conditions with any pair of probability measures with finite moments of order $2$.  In parallel, considering the weights

\begin{gather}
\begin{align*}
W^{0}_i=W^{0}_{{i:n_0}}/\sum_{i=1}^{n_0} W^{0}_{{i:n_0}} \hspace{0.2cm} and \hspace{0.2cm} W^{1}_j=W^{1}_{{j:n_1}}/\sum_{i=j}^{n_1} W^{0}_{{j:n_1}} \hspace{0.2cm} (i=1,\dots,n_0) \hspace{0.2cm } (j=1,\dots,n_1),
\end{align*}	
\end{gather}

we have

\begin{gather}
\begin{align*}
 W^{0}_i\geq 0, \hspace{0.2cm} W^{1}_j	\geq0 \hspace{0.2cm} (i=1,\dots,n_0) \hspace{0.2cm } (j=1,\dots,n_1) \hspace{0.2cm} \sum_{i=1}^{n_0} W^{0}_i=1 \hspace{0.2cm} \text{and} \hspace{0.2cm} \sum_{j=1}^{n_1} W^{1}_j=1.
\end{align*}
\end{gather}

Now, we consider the probability measures $P^{*}_{0}$, $P^{*}_{1}$ induced by the probabilities $(W^{0}_1,\dots, W^{0}_{n_0})$, $(W^{1}_1,\dots, W^{1}_{n_1})$  whose values are  $(X_{01},\dots, X_{0n_{0}})$ and $(X_{11},\dots, X_{1n_1})$ respectively. It is trivially verified that the energy distance and the maximum mean discrepancy between $P^{*}_{0}$ and $P^{*}_1$ are well defined. By definition,

\begin{align}\label{energiaapend}
\epsilon_\alpha(P^{*},Q^{*})= 2E||X-Y||^{\alpha}- E||X-X^{'}||^{\alpha}-E||Y-Y^{'}||^{\alpha}
\end{align}

where $X$,$X^{\prime}\sim^{i.i.d. } P^{*}$ and $Y$,$Y^{\prime}\sim^{i.i.d.}Q^{*}$. 

Replacing \ref{energiaapend} with  the  populations defined above  quantities,

\begin{gather}
\begin{align}
\epsilon_{\alpha}(P_0^{*},P_1^{*})= 2  \frac{\sum_{i=1}^{n_0}  \sum_{j=1}^{n_1} W^{0}_{i:n_0} W^{1}_{i:n_1} ||X_{0(i:n_0)}-X_{1(j:n_1)}||^{\alpha}}{\sum_{i=1}^{n_0}  \sum_{j=1}^{n_1} W^{0}_{i:n_0} W^{1}_{j:n_1}} - \frac{\sum_{i=1}^{n_0} \sum_{j\neq i}^{n_0}  W^{0}_{i:n_0} W^{0}_{j:n_0}  ||X_{0(i:n_0)}-X_{0(j:n_0)}||^{\alpha}}{\sum_{i=1}^{n_0} \sum_{j= i}^{n_0} W^{0}_{i:n_0} W^{0}_{j:n_0}} \\
- \frac{\sum_{i=1}^{n_1}  \sum_{i\neq j}^{n_1} W^{1}_{i:n_1} W^{1}_{j:n_1}  ||X_{1(i:n_1)}-X_{1(j:n_1)}||^{\alpha}}{\sum_{i=1}^{n_1}  \sum_{j= i}^{n_1} W^{1}_{i:n_1} W^{1}_{j:n_1}} \nonumber.
\end{align}
\end{gather}

 Therefore equations \ref{eqn:estadistico3} and \ref{eqn:estadistico4}  always take values greater than or
equal to zero. This is given if and only if $(W^{0}_1,\dots, W^{0}_{n_0})=(W^{1}_1,\dots, W^{1}_{n_1})$ and $(X_{01},\dots, X_{0n_{0}})=(X_{11},\dots, X_{1n_1})$. This also implies  $(\delta_{01},\dots, \delta_{0n_{0}})=(\delta_{11},\dots, \delta_{1n_1})$.

Note that it is well known that the $U$-statistics  does not verify that property in the general case.  The same is true in the case of censorship present.

\section{Construction of the statistics and mathematical meaning of limits}
\label{appendix:C}

Let $\tilde{\epsilon}_{\alpha}(P_0,P_1)$ be the statistic of the $\alpha$ energy distance   without normalizing (see expression (\ref{sin normalizarr})).  It can be proved by the following reasoning, similar to the appendix \ref{appendix:proof}:  

\begin{gather*}\label{limitesinnormalizar}
\begin{align}
\tilde{\epsilon}_{\alpha}(P_0,P_1)\overset{n_0,n_1\to \infty}{\to} \epsilon_{c(\alpha)}(P_0,P_1)= 2\int_{0}^{\tau_0} \int_{0}^{\tau_1}	 ||x-y||^{\alpha} dP_0(x) dP_1(y)\\
-\int_{0}^{\tau_0} \int_{0}^{\tau_0}	 ||x-y||^{\alpha} dP_0(x) dP_0(y) -\int_{0}^{\tau_1}\int_{0}^{\tau_1}	 ||x-y||^{\alpha} dP_1(x) dP_1(y). 
\end{align}
\end{gather*}

Now, we consider $P_0$ to be the distribution function of a random variable $N(100000,1)$ and $P_1$ to be  a $Uniform(0,1)$ and $\tau_0=\tau_1=0.1$. Then,  $\epsilon_{c(\alpha)}(P_0,P_1)<0$. Next, we  define a function  $f:\mathbb{R}^{+}\to \mathbb{R}$ that evaluates $\epsilon_{c(\alpha)}(P_0,P_1)$ where $P_1$ is defined as before and $P_0$ is the distribution function of a random variable $N(\mu,1)$,  $\mu$ being the parameter of the function. Using the Bolzano Theorem, we can see that there exists $\mu^{*}\in \mathbb{R}^{+}$ such that $f(\mu^{*})=0$. However, (\ref{sin normalizarr}) does not define a consistent test against all alternatives because $\exists t\in [0,0.1]$ such that $P_0(t)\neq P_1(t)$. In fact, it can be proved in this example that it may have only one point where both distribution functions coincide.

The foregoing shows that for  energy distance to be positive, it must be evaluated on a measure of probability in its complete range. This naturally leads to the standardization of statistics (\ref{eqn:estadistico1}--\ref{eqn:estadistico4}) as in the case of censored data. Another condition for   energy distance to behave as a true distance between probability distributions is that the semimetric is of a negative type. Every metric defined in a Hilbert space verifies that condition and therefore the usual Euclidean distance guarantees that property. \cite{lyons2013distance} present a deeper discussion about the related aspects.

In its abstract version, the energy distance between two distribution functions $P_0$, $P_1$ does not have an interpretable explicit expression for distribution functions and characteristic functions of the random variables  involved. \cite{lyons2013distance} precisely extended the energy distance to metric spaces without using Fourier analysis. However, in the case of using the Euclidean metric or  an invariant kernel such as the Gaussian kernel, we can give an explicit expression at the population level. Below, we provide  concrete expressions for energy distance with euclidean distance and maximun mean discrepancy with the Gaussian kernel:

\begin{gather}
\begin{align*}
\epsilon_{c_{(1)}}(P_0,P_1)=  2\int_{0}^{\infty} (P^{\prime}_0(t)-P^{\prime}_1(t))^{2} dt= \\
	2\frac{\int_{0}^{\tau_0} \int_{0}^{\tau_1} ||x-y|| dP^{0}(x) dP^{1}(y)}{\int_{0}^{\tau_0}dP^{0}(x)\int_{0}^{\tau_1}dP^{1}(x)}-\frac{\int_{0}^{\tau_0} \int_{0}^{\tau_0} ||x-y|| dP^{0}(x) dP^{0}(y)}{\int_{0}^{\tau_0}dP^{0}(x)\int_{0}^{\tau_0}dP^{0}(x)}-\frac{\int_{0}^{\tau_1} \int_{0}^{\tau_1} ||x-y|| dP^{1}(x) dP^{1}(y)}{\int_{0}^{\tau_1}dP^{1}(x)\int_{0}^{\tau_1}dP^{1}(x)}.
\end{align*}
\end{gather}

\begin{gather}
\begin{align*}
\epsilon_c(\alpha)(P_0,P_1)=  \frac{1}{\pi}\int_{0}^{\infty} \frac{|\hat{f_0}(t)-\hat{f}_1(t)|^2}{|t|^{2}} dt
\end{align*}
\end{gather}

where $\hat{f}_0$ is the characteristic function of $P_0^{\prime}$ and 
$\hat{f}_1$ is the characteristic function of $P_1^{\prime}$.

Finally, given a Gaussian kernel $K$ or any translation invariant kernels

\begin{gather}
\begin{align*}
\gamma_{c(K)}^{2}(P_0,P_1) = \int_{0}^{t} |\hat{f_0}(t)-\hat{f}_1(t)|^2 d\Lambda(t)
\end{align*}  
\end{gather}

where $\Lambda(\cdot)$ is a finite non-negative Borel measure.

\section{Statistics in multivariate case}

\label{appendix:D}



Let us now consider the construction of the statistics of energy distance and maximum mean discrepancy in the multivariate case. In this case, there is a  lifetime $T\in \mathbb{R}^{+}$ with possible censorship  and a vector of covariates $S\in \mathbb{R}^{p-1}$ without censorship. Possible practical applications of the above include the comparison of the equality of distribution according to the lifetime of individuals and certain clinical variables of patients, independence testing \cite{shen2019exact},  or change-point detection problems.

Let $H_{j,i}=(T_{j,i},S_{j,i}) \sim P_j$ $(j=0,1; i=1,\dots,n_j)$ and censoring times $C_{j,i} \sim$ $Q_j$ $(j=0,1; i=1,\dots,n_j)$, with distribution   $P_j$  defined as a subset of $\mathbb{R^{+}}\times \mathbb{R}^{p-1}$ and the distributions  $Q_j$ on $\mathbb{R}^{+}$ $(j=0,1)$. Here,  the index   $j$  represents  a population, and the index  $i$  a particular  sample within the corresponding  population. Moreover, the random variables  $(T_{0,1},S_{0,1}),\dots,(T_{0,n_0},S_{0,n_0}),\dots,(T_{1,1},S_{1,1}),$ $\dots,(T_{1,n_1},S_{1,n_1}),$ $ C_{0,1},\dots,C_{n_0,n_0},C_{1,n_1},\dots,C_{n_1,n_1}$ are assumed to be independent of each other. In practice, only the random variables   $(X_{j,i}= min(T_{j,i},C_{j,i}),S_{j,i})$  and $\delta_{j,i}= 1\{X_{j,i} = T_{j,i}\}$	$(j=0,1; i=1,\dots,n_j)$ are observed. 

On the basis  of the  observed  data $\{(X_{j,i},S_{j,i},\delta_{j,i})\}_{j=0,1; i=1,\dots,n_j}$ we must approximate the distances $\epsilon_\alpha(P_0,P_1)$, ${\gamma}_K^{2}(P_0,P_1)$. In this case, we can use the Kaplan-Meier estimator in the presence of covariates \cite{stute1993consistent,gerds2017kaplan}.

\begin{gather*}
\begin{align}
\hat{\epsilon}_{\alpha}(P_0,P_1)= 2  \frac{\sum_{i=1}^{n_0}  \sum_{j=1}^{n_1} W^{0}_{i:n_0} W^{1}_{i:n_1} ||H_{0,(i:n_0)}-H_{1,(j:n_1)}||^{\alpha}}{\sum_{i=1}^{n_0}  \sum_{j=1}^{n_1} W^{0}_{i:n_0} W^{1}_{j:n_1}} - \frac{\sum_{i=1}^{n_0} \sum_{j\neq i}^{n_0}  W^{0}_{i:n_0} W^{0}_{j:n_0}  ||H_{0,(i:n_0)}-H_{0,(j:n_0)}||^{\alpha}}{\sum_{i=1}^{n_0} \sum_{j\neq i}^{n_0} W^{0}_{i:n_0} W^{0}_{j:n_0}} \\
- \frac{\sum_{i=1}^{n_1}  \sum_{i\neq j}^{n_1} W^{1}_{i:n_1} W^{1}_{j:n_1}  ||H_{1,(i:n_1)}-H_{1,(j:n_1)}||^{\alpha}}{\sum_{i=1}^{n_1}  \sum_{j\neq i}^{n_1} W^{1}_{i:n_1} W^{1}_{j:n_1}} \nonumber
\end{align}
\end{gather*}
\textbf{($U$-statistic $\alpha$-energy distance under right censoring)},

\begin{gather}
\begin{align*}
\hat{\gamma}_K^{2}(P_0,P_1)=  \frac{\sum_{i=1}^{n_0}  \sum_{j\neq i}^{n_0} W^{0}_{i:n_0} W^{0}_{j:n_0} K(H_{0,(i:n_0)},H_{0,(j:n_0)})}{ \sum_{j=1}^{n_0} \sum_{j\neq i}^{n_0} W^{0}_{i:n_0} W^{0}_{j:n_0}}+ \frac{\sum_{i=1}^{n_1}  \sum_{j\neq i}^{n_1} W^{1}_{i:n_1} W^{1}_{j:n_1}  K(H_{1,(i:n_1)},H_{1,(j:n_1)})}{\sum_{i=1}^{n_1}  \sum_{j\neq i }^{n_1}W^{1}_{i:n_1} W^{1}_{j:n_1}}\\-2\frac{\sum_{i=1}^{n_0}  \sum_{j=1}^{n_1} W^{0}_{i:n_0} W^{1}_{i:n_1} K(H_{0,(i:n_0)},H_{1,(j:n_1)})}{\sum_{i=1}^{n_0}  \sum_{j=1}^{n_1} W^{0}_{i:n_0} W^{1}_{i:n_1}} \nonumber
\end{align*}
\end{gather}

\textbf{($U$-statistic kernel method under right censoring)}.

Analogously, we can define $V$-statistics as follows:

\begin{gather}
\begin{align*}
\hat{\epsilon}_{\alpha}(P_0,P_1)= 2  \frac{\sum_{i=1}^{n_0}  \sum_{j=1}^{n_1} W^{0}_{i:n_0} W^{1}_{i:n_1} ||H_{0,(i:n_0)}-H_{1,(j:n_1)}||^{\alpha}}{\sum_{i=1}^{n_0}  \sum_{j=1}^{n_1} W^{0}_{i:n_0} W^{1}_{i:n_1}} - \frac{\sum_{i=1}^{n_0} \sum_{j=1 }^{n_0}  W^{0}_{i:n_0} W^{0}_{j:n_0}  ||H_{0,(i:n_0)}-H_{0,(j:n_0)}||^{\alpha}}{\sum_{i=1}^{n_0} \sum_{j= 1}^{n_0} W^{0}_{i:n_0} W^{0}_{j:n_0}} \\
- \frac{\sum_{i=1}^{n_1}  \sum_{j=1}^{n_1} W^{1}_{i:n_1} W^{1}_{j:n_1}  ||H_{1,(i:n_1)}-H_{1,(j:n_1)}||^{\alpha}}{\sum_{i=1}^{n_1}  \sum_{j=1}^{n_1} W^{1}_{i:n_1} W^{1}_{j:n_1}} \nonumber
\end{align*}
\end{gather}

\textbf{($V$-statistic $\alpha$-energy distance under right censoring)},
\begin{gather}
\begin{align*}
\hat{\gamma}_K^{2}(P_0,P_1)=  \frac{\sum_{i=1}^{n_0}  \sum_{j=1}^{n_0} W^{0}_{i:n_0} W^{0}_{j:n_0} K(H_{0,(i:n_0)},H_{0,(j:n_0)})}{ \sum_{j=1}^{n_0} \sum_{j=1}^{n_0} W^{0}_{i:n_0} W^{0}_{j:n_0}}+ \frac{\sum_{i=1}^{n_1}  \sum_{j=1}^{n_1} W^{1}_{i:n_1} W^{1}_{j:n_1}  K(H_{1,(i:n_1)},H_{1,(j:n_1)})}{\sum_{i=1}^{n_1}  \sum_{j=1 }^{n_1}W^{1}_{i:n_1} W^{1}_{j:n_1}}\\-2\frac{\sum_{i=1}^{n_0}  \sum_{j=1}^{n_1} W^{0}_{i:n_0} W^{1}_{i:n_1} K(H_{0,(i:n_0)},H_{1,(j:n_1)})}{\sum_{i=1}^{n_0}  \sum_{j=1}^{n_1} W^{0}_{i:n_0} W^{1}_{i:n_1}} \nonumber
\end{align*}
\end{gather}

\textbf{($V$-statistic kernel method under right censoring)},
where
\begin{equation}
W^{0}_{{i:n_0}}= \frac{\delta_{0,(i:n_0)}}{n_0-i+1}\prod_{j=1}^{i-1}[\frac{n_0-j}{n_0-j+1}]^{\delta_{0,(j:n_0)}}
\hspace{0.2cm}   (i=1,\dots,n_0)
\end{equation}
and
\begin{equation}
W^{1}_{{i:n_1}}= \frac{\delta_{1,(i:n_1)}}{n_1-i+1}\prod_{j=1}^{i-1}[\frac{n_1-j}{n_1-j+1}]^{\delta_{1,(j:n_1)}} \hspace{0.2cm}   (i=1,\dots,n_1).
\end{equation}

It can be seen that this estimator is asymptotically efficient with the hypothesis of independence assumed between lifetimes and censorship times \cite{gerds2017kaplan}.  However, this situation is unrealistic in practice.  	Instead, $T$ and $C$ are often imposed to be  conditionally independent given   $S$, see \cite{fan1994censored}. 

Given the equivalence between the weights of the Kaplan-Meier estimator and the inverse-probability-of-censoring weighted average 
\cite{satten2001kaplan}, a natural generalization for modeling dependent censorship is to calculate weights as follows:

\begin{equation}
W^{1}_{{i:n_0}}= \frac{\delta_{0(i:n_1)}}{n_0\hat{P}(C_0>X_{0,(i:n_0)}|S=S_{0,(i:n_0)})} \hspace{0.2cm}   (i=1,\dots,n_0),
\end{equation}
and
\begin{equation}
W^{1}_{{i:n_1}}= \frac{\delta_{1(i:n_1)}}{n_1 \hat{P}(C_1>X_{1,(i:n_1)}|S=S_{1,(i:n_i)})} \hspace{0.2cm}   (i=1,\dots,n_1).
\end{equation}.

The previous conditional probability of the censorship variable of each population can be estimated, for example, using the Cox model \cite{gerds2017kaplan}. In   a one- or two-dimensional space, an alternative option is to use a non-parametric approach with the Beran estimator (the smoothed conditional Kaplan-Meier estimator) \cite{gonzalez1994asymptotic}. From the theoretical point of view, in the case of dependent censorship, the estimators with inverse-probability-of-censoring weighted average   have the disadvantage that they are  not  asymptotically efficient  \cite{van2002locally}. A doubly robust strategy \cite{tsiatis2007semiparametric,rubin2007doubly} could solve this problem, however this is an open problem for high-dimensional data.

\section{Null hypothesis results}
\label{appendix:E}

\begin{table}[H]

	\caption{	\label{append:tab:tabla3} Proportion $p$-values less or equal $0.05$ Exponencial distribution under the null hypothesis}
	\scalebox{0.8}{

		\begin{tabular}{ccccccccccc}			\hline
			Method: & &	& &	 Logrank & Gehan & Tarone & Peto & Flemming \\
			& &	&  & & & & & $\rho=1,\gamma=1$  \\

			Comparative	& $n_1$ & $n_2$ & Censoring rate    & $\hat{p}$  & $\hat{p}$ & $\hat{p}$ & $\hat{p}$ & $\hat{p}$	
			
			\\

			\hline
			Exp(1) & 20 & 20 & 0.1  & 0.050 & 0.046 & 0.046 & 0.044 & 0.058 \\ 
			Exp(1) & 50 & 50 & 0.1  & 0.060 & 0.066 & 0.064 & 0.064 & 0.056 \\ 
			Exp(1.5) & 20 & 20  & 0.1 & 0.066 & 0.058 & 0.062 & 0.058 & 0.062 \\ 
			Exp(1.5) & 50 & 50 & 0.1  & 0.062 & 0.056 & 0.050 & 0.056 & 0.054 \\ 
			Exp(1) & 20 & 20 & 0.3   & 0.058 & 0.054 & 0.056 & 0.056 & 0.056 \\ 
			Exp(1) & 50 & 50 &  0.3  & 0.050 & 0.052 & 0.054 & 0.058 & 0.046 \\ 
			Exp(1.5) & 20 & 20 & 0.3  & 0.054 & 0.054 & 0.050 & 0.052 & 0.052 \\ 
			Exp(1.5) & 50 & 50 & 0.3 &  0.066 & 0.068 & 0.066 & 0.068 & 0.056 \\
			\hline
		\end{tabular}} 
		
		\scalebox{0.8}{
			\begin{tabular}{ccccccccc}			\hline
				Method: & &	& &	 Energy distance & Kernel & Kernel  \\
				& &	& &	 $\alpha=1$  & Gaussian $\sigma=1$  & Laplacian $\sigma=1$   \\
				
				Comparative	& $n_1$ & $n_2$ & Censoring rate   &  $\hat{p}$ & $\hat{p}$ & $\hat{p}$ 
				
				\\

				\hline
				Exp(1) & 20 & 20 & 0.1 & 0.048 &  0.052 & 0.046 &  \\ 
				Exp(1) & 50 & 50 & 0.1 & 0.056  & 0.052 & 0.056   \\ 
				Exp(1.5) & 20 & 20  & 0.1 &  0.066  & 0.072 & 0.058  \\ 
				Exp(1.5) & 50 & 50 & 0.1 & 0.042  & 0.048 & 0.044 \\ 
				Exp(1) & 20 & 20 & 0.3 & 0.058  & 0.050 & 0.042  \\ 
				Exp(1) & 50 & 50 &  0.3 & 0.056 & 0.054 & 0.050  \\ 
				Exp(1.5) & 20 & 20 & 0.3 & 0.058 & 0.052 & 0.052   \\ 
				Exp(1.5) & 50 & 50 & 0.3 & 0.064 & 0.064 & 0.044  \\
				\hline
			\end{tabular}} 
			\label{fig:n11}	
		\end{table}

		\begin{table}[H]

			\caption{				\label{append:tab:tabla4}  Proportion $p$-values less or equal $0.05$ Gamma distribution nunder the null hypothesis}
			\centering
			\scalebox{0.8}{
				\begin{tabular}{ccccccccccc}			\hline
					Method: & &	& &	  Logrank & Gehan & Tarone & Peto & Flemming \\
					& &	& &	   & & & & $\rho=1,\gamma=1$  \\

					Comparative	& $n_1$ & $n_2$ & Censoring rate   & $\hat{p}$ &  $\hat{p}$ & $\hat{p}$ & $\hat{p}$ & $\hat{p}$	
					
					\\

					\hline 
					Gamma(1,1) & 20 & 20 & 0.3  & 0.054 & 0.054 & 0.056 & 0.058 & 0.054 \\ 
					Gamma(1,1) & 50 & 50 & 0.1 & 0.038 & 0.038 & 0.030 & 0.032 & 0.050 \\ 
					Gamma(1.5,1.5) & 20 & 20 & 0.1   & 0.046 & 0.048 & 0.048 & 0.048 & 0.062 \\ 
					Gamma(1.5,1.5) & 50 & 50 & 0.1  & 0.046 & 0.044 & 0.046 & 0.044 & 0.050 \\ 
					Gamma(1,1) & 20 & 20 & 0.3 &  0.056 & 0.060 & 0.058 & 0.052 & 0.066 \\ 
					Gamma(1,1) & 50 & 50 & 0.3 & 0.054 & 0.052 & 0.058 & 0.050 & 0.046 \\ 
					Gamma(1.5,1.5) & 20 & 20 & 0.3  & 0.058 & 0.060 & 0.064 & 0.062 & 0.066 \\ 
					Gamma(1.5,1.5) & 50 & 50 & 0.3  & 0.050 & 0.062 & 0.060 & 0.062 & 0.050 \\ 
					
					\hline
				\end{tabular}} 
				
				\scalebox{0.8}{
					\begin{tabular}{ccccccccc}			\hline
						Method: & &	& &	 Energy distance  & Kernel & Kernel  \\
						& &	& &	 $\alpha=1$  & Gaussian $\sigma=1$  & $Laplacian$ $\sigma=1$    \\

						Comparative	& $n_1$ & $n_2$ & Censoring rate   & $\hat{p}$ & $\hat{p}$ & $\hat{p}$ & 	
						
						\\

						\hline 
						Gamma(1,1) & 20 & 20 & 0.3 & 0.058  & 0.052 & 0.060   \\ 
						Gamma(1,1) & 50 & 50 & 0.1 & 0.044 & 0.042 & 0.040   \\ 
						Gamma(1.5,1.5) & 20 & 20 & 0.1  & 0.062  & 0.060 & 0.060  \\ 
						Gamma(1.5,1.5) & 50 & 50 & 0.1 & 0.050 & 0.054 & 0.052  \\ 
						Gamma(1,1) & 20 & 20 & 0.3 & 0.058  & 0.064  & 0.062  \\ 
						Gamma(1,1) & 50 & 50 & 0.3 & 0.058  & 0.056 & 0.062  \\ 
						Gamma(1.5,1.5) & 20 & 20 & 0.3  & 0.068 & 0.070 & 0.056 \\ 
						Gamma(1.5,1.5) & 50 & 50 & 0.3 &  0.056 & 0.056 & 0.068  \\ 
						\hline
						
					\end{tabular}} 
					\label{fig:n11}	
					
				\end{table}

				\begin{table}[H]

						\caption{						\label{append:tab:tabla5} Proportion $p$-values less or equal $0.05$ Lognormal distribution  under the null hypothesis}
					\centering
					\scalebox{0.8}{
						\begin{tabular}{ccccccccccc}			\hline
							Method: & &	& &	  Logrank & Gehan & Tarone & Peto & Flemming \\
							& &	& &	   & & & & $\rho=1,\gamma=1$  \\

							Comparative	& $n_1$ & $n_2$ & Censoring rate   & $\hat{p}$ &  $\hat{p}$ & $\hat{p}$ & $\hat{p}$ & $\hat{p}$	
							\\ 
							
							\hline 
							Lognormal(0,0.5) & 20 & 20 & 0.1  & 0.052 & 0.044 & 0.044 & 0.042 & 0.048 \\ 
							Lognormal(0,0.5) & 50 & 50 & 0.1   & 0.040 & 0.034 & 0.040 & 0.036 & 0.040 \\ 
							Lognormal(0,0.25) & 20 & 20 & 0.1  & 0.062 & 0.078 & 0.076 & 0.080 & 0.054 \\ 
							Lognormal(0,0.25) & 50 & 50 & 0.1   & 0.036 & 0.044 & 0.044 & 0.040 & 0.038 \\ 
							Lognormal(0,0.5) & 20 & 20 & 0.3  & 0.042 & 0.052 & 0.040& 0.048 & 0.050 \\ 
							Lognormal(0,0.5) & 50 & 50  & 0.3  & 0.078 & 0.082 & 0.078 & 0.082 & 0.066 \\ 
							Lognormal(0,0.25) & 20 & 20 & 0.3 & 0.050 & 0.056 & 0.058 & 0.054 & 0.042 \\ 
							Lognormal(0,0.25) & 50 & 50 & 0.3  & 0.046 & 0.060 & 0.046 & 0.058 & 0.048 \\ 
							\hline
						\end{tabular}}

						\scalebox{0.8}{
							\begin{tabular}{ccccccccc}			\hline
								Method: & &	& &	 Energy distance  & Kernel & Kernel  \\
								& &	& &	 $\alpha=1$  & Gaussian $\sigma=1$  & $Laplacian$ $\sigma=1$    \\

								Comparative	& $n_1$ & $n_2$ & Censoring rate   & $\hat{p}$ & $\hat{p}$ & $\hat{p}$ & 	
								
								\\

								\hline 
								
								Lognormal(0,0.5) & 20 & 20 & 0.1 & 0.050 & 0.052 & 0.054  \\ 
								Lognormal(0,0.5) & 50 & 50 & 0.1 & 0.040  & 0.038 & 0.040   \\ 
								Lognormal(0,0.25) & 20 & 20 & 0.1 & 0.084  & 0.076 & 0.080  \\ 
								Lognormal(0,0.25) & 50 & 50 & 0.1 & 0.038 & 0.040 & 0.044   \\ 
								Lognormal(0,0.5) & 20 & 20 & 0.3 & 0.046  & 0.050 & 0.046  \\ 
								Lognormal(0,0.5) & 50 & 50 & 0.3 & 0.072 & 0.074 & 0.074   \\ 
								Lognormal(0,0.25) & 20 & 20 & 0.3 & 0.056 & 0.060 & 0.054   \\ 
								Lognormal(0,0.25) & 50 & 50 & 0.3 & 0.044 & 0.040 & 0.052  \\ 
								\hline
							\end{tabular}}

							\label{fig:n11}	
							
						\end{table}

\begin{table}[H]
	\caption{	\label{append:tab:tabla6} Empirical mean and standart deviation of $p-values$ Exponential distribution under the null hyphotesis. }

	\scalebox{0.8}{
		\begin{tabular}{ccccccccccc}			
			\hline
			Method: & &	& &	 Logrank & Gehan & Tarone & Peto & Flemming \\
			& &	& &	  & & & & $\rho=1,\gamma=1$  \\
			Comparative	& $n_1$ & $n_2$ & Censoring rate   & $\overline{x}\Mypm \sigma$ & $\overline{x}\Mypm \sigma$ & $\overline{x}\Mypm \sigma$ & $\overline{x}\Mypm \sigma$ &  $\overline{x}\Mypm \sigma$	
			\\
			\hline	
			Exp(1)	& 50 & 50 & 0.1  & 0.492 $\Mypm$   0.293 & 0.489 $\Mypm$   0.295 & 0.478 $\Mypm$   0.280 & 0.486 $\Mypm$   0.292 & 0.490 $\Mypm$   0.295 \\ 
			Exp(1.5)	& 20 & 20	 & 0.1 & 0.486 $\Mypm$   0.289 & 0.475 $\Mypm$   0.291 & 0.481 $\Mypm$   0.288 & 0.458 $\Mypm$   0.288 \\ 
			Exp(1.5)	& 50 & 50	 & 0.1 & 0.492 $\Mypm$   0.299 & 0.501 $\Mypm$   0.295 & 0.492 $\Mypm$   0.295 & 0.498 $\Mypm$   0.295 & 0.479 $\Mypm$   0.293 \\ 
			Exp(1)	& 20 & 20	 & 0.3  & 0.495 $\Mypm$   0.284 & 0.502 $\Mypm$   0.297 & 0.500 $\Mypm$   0.294 & 0.499 $\Mypm$   0.295 & 0.507 $\Mypm$   0.286 \\ 
			Exp(1)	& 50 & 50	 & 0.3  & 0.503 $\Mypm$   0.296 & 0.486 $\Mypm$   0.297 & 0.491 $\Mypm$   0.296 & 0.486 $\Mypm$   0.297 & 0.502 $\Mypm$   0.287 \\ 
			Exp(1.5)	& 20 & 20	 & 0.3   & 0.497 $\Mypm$   0.295 & 0.499 $\Mypm$   0.289 & 0.493 $\Mypm$   0.285 & 0.495 $\Mypm$   0.286 & 0.501 $\Mypm$   0.292 \\ 
			Exp(1.5)	& 50 & 50	 & 0.3 & 0.496 $\Mypm$   0.298 & 0.492 $\Mypm$   0.294 & 0.495 $\Mypm$   0.299 & 0.492 $\Mypm$   0.294 & 0.500 $\Mypm$   0.299 \\ 
			\hline
		\end{tabular}}

		\scalebox{0.8}{

			\begin{tabular}{ccccccccc}
				\hline
				Method: & &	& &	 Energy distance   & Kernel & Kernel  \\
				& &	& &	 $\alpha=1$ & Gaussian $\sigma=1$  & $Laplacian$ $\sigma=1$    \\
				Comparative	& $n_1$ & $n_2$ & Censoring rate   & $\overline{x}\Mypm \sigma$ & $\overline{x}\Mypm \sigma$ &  $\overline{x}\Mypm \sigma$
				\\						
				\hline	
				Exp(1)	& 50 & 50 & 0.1 & 0.482 $\Mypm$   0.293 &  0.481 $\Mypm$   0.298 & 0.478 $\Mypm$   0.294   \\ 
				Exp(1.5)	& 20 & 20	 & 0.1 & 0.482 $\Mypm$   0.287  & 0.490 $\Mypm$   0.285 & 0.493 $\Mypm$   0.288   \\ 
				Exp(1.5)	& 50 & 50	 & 0.1 & 0.482 $\Mypm$   0.295  & 0.485 $\Mypm$   0.293 & 0.481 $\Mypm$   0.289  \\ 
				Exp(1)	& 20 & 20	 & 0.3 & 0.508 $\Mypm$   0.288  & 0.503 $\Mypm$   0.285 & 0.506 $\Mypm$   0.287  \\ 
				Exp(1)	& 50 & 50	 & 0.3 & 0.494 $\Mypm$   0.297  & 0.493 $\Mypm$   0.297 & 0.495 $\Mypm$   0.297   \\ 
				Exp(1.5)	& 20 & 20	 & 0.3 & 0.500 $\Mypm$   0.290  & 0.492 $\Mypm$   0.284 & 0.506 $\Mypm$   0.292 \\ 
				Exp(1.5)	& 50 & 50	 & 0.3 & 0.489 $\Mypm$   0.301 & 0.489 $\Mypm$   0.301  & 0.490 $\Mypm$   0.302 \\ 
				\hline	
			\end{tabular}}
		\end{table}

		\begin{table}[H]

			\caption{\label{append:tab:tabla7} Empirical mean and standart deviation of $p-values$ Gamma distribution under the null hyphotesis.}
			
			\scalebox{0.8}{
				\begin{tabular}{ccccccccccc}			\hline
					Method: & &	& &	 Logrank & Gehan & Tarone & Peto & Flemming \\
					& &	& &	  & & & & $\rho=1,\gamma=1$  \\
					Comparative	& $n_1$ & $n_2$ & Censoring rate   & $\overline{x}\Mypm \sigma$ & $\overline{x}\Mypm \sigma$ & $\overline{x}\Mypm \sigma$ & $\overline{x}\Mypm \sigma$ &  $\overline{x}\Mypm \sigma$
					
					\\
					
					\hline	
					
					Gamma(1,1)	& 20  & 20 & 0.1  & 0.491 $\Mypm$   0.284 & 0.510 $\Mypm$   0.294 & 0.498 $\Mypm$   0.288 & 0.506 $\Mypm$   0.292 & 0.493 $\Mypm$   0.282 \\ 
					Gamma(1,1)	& 50  & 50 & 0.1&  0.505 $\Mypm$   0.292 & 0.508 $\Mypm$   0.287 & 0.505 $\Mypm$   0.290 & 0.508 $\Mypm$   0.288 & 0.502 $\Mypm$   0.290 \\ 
					Gamma(1.5,1.5)	& 20  & 20 & 0.1 &  0.515 $\Mypm$   0.299 &  $\Mypm$   0.299 & 0.520 $\Mypm$   0.290 & 0.519 $\Mypm$   0.289 & 0.522 $\Mypm$   0.291   \\ 
					Gamma(1.5,1.5)	& 50  & 50 & 0.1&  0.493 $\Mypm$   0.291 &  0.505 $\Mypm$   0.289 & 0.509 $\Mypm$   0.289 & 0.506 $\Mypm$   0.288 & 0.505 $\Mypm$   0.291 \\ 
					Gamma(1,1)	& 20  & 20 & 0.3 &  0.484 $\Mypm$   0.288 & 0.475 $\Mypm$   0.289 & 0.477 $\Mypm$   0.297 & 0.467 $\Mypm$   0.288 & 0.464 $\Mypm$   0.288 &  \\ 
					Gamma(1,1)	& 50  & 50 & 0.3 &  0.485 $\Mypm$   0.292 & 0.513 $\Mypm$   0.300 & 0.498 $\Mypm$   0.293 & 0.511 $\Mypm$   0.300 & 0.474 $\Mypm$   0.287 \\  
					Gamma(1.5,1.5)	& 20  & 20 & 0.3 & 0.484 $\Mypm$   0.297 & 0.499 $\Mypm$   0.294 & 0.490 $\Mypm$   0.295 & 0.494 $\Mypm$   0.292 & 0.484 $\Mypm$   0.294 \\ 
					Gamma(1.5,1.5)	& 50  & 50 & 0.3 & 0.509 $\Mypm$   0.289 & 0.490 $\Mypm$  0.291 & 0.493 $\Mypm$   0.288 & 0.489 $\Mypm$   0.292 & 0.514 $\Mypm$   0.288 \\ 
					
					\hline

				\end{tabular}}

				\scalebox{0.8}{
					\begin{tabular}{ccccccccc}			\hline
						Method: & &	& &	 Energy distance   & Kernel & Kernel  \\
						& &	& &	 $\alpha=1$ & Gaussian $\sigma=1$  & $Laplacian$ $\sigma=1$    \\
						Comparative	& $n_1$ & $n_2$ & Censoring rate   & $\overline{x}\Mypm \sigma$ & $\overline{x}\Mypm \sigma$ &  $\overline{x}\Mypm \sigma$
						\\						
						\hline
						Gamma(1,1)	& 20  & 20 & 0.1 & 0.501 $\Mypm$   0.294  & 0.512 $\Mypm$   0.297 & 0.508 $\Mypm$   0.296   \\ 
						Gamma(1,1)	& 50  & 50 & 0.1& 0.503 $\Mypm$   0.291 & 0.512 $\Mypm$   0.292 & 0.508 $\Mypm$   0.288  \\ 
						Gamma(1.5,1.5)	& 20  & 20 & 0.1 & 0.519 $\Mypm$   0.295 & 0.515 $\Mypm$   0.301 & 0.516 $\Mypm$   0.295   \\ 
						Gamma(1.5,1.5)	& 50  & 50 & 0.1&  0.499 $\Mypm$   0.290 & 0.493 $\Mypm$   0.291  & 0.495 $\Mypm$   0.292  \\ 
						Gamma(1,1)	& 20  & 20 & 0.3  & 0.477 $\Mypm$   0.288 & 0.479 $\Mypm$   0.289 & 0.484 $\Mypm$   0.287 \\ 
						Gamma(1,1)	& 50  & 50 & 0.3	 & 0.489 $\Mypm$   0.293 & 0.497 $\Mypm$   0.296 & 0.497 $\Mypm$   0.293  \\  
						Gamma(1.5,1.5)	& 20  & 20 & 0.3 & 0.491 $\Mypm$   0.293 & 0.493 $\Mypm$   0.294 & 0.494 $\Mypm$   0.294  \\ 
						Gamma(1.5,1.5)	& 50  & 50 & 0.3	& 0.495 $\Mypm$   0.295  & 0.492 $\Mypm$   0.293 & 0.490 $\Mypm$   0.295   \\ 
						\hline
					\end{tabular}}

				\end{table}

				\begin{table}[H]

					\caption{					\label{append:tab:tabla8} Empirical mean and standart deviation of $p-values$ Lognormal distribution under the null hyphotesis.}

					\scalebox{0.8}{
						\begin{tabular}{ccccccccccc}			\hline
							Method: & &	& &	 Logrank & Gehan & Tarone & Peto & Flemming \\
							& &	& &	  & & & & $\rho=1,\gamma=1$  \\
							Comparative	& $n_1$ & $n_2$ & Censoring rate   & $\overline{x}\Mypm \sigma$ & $\overline{x}\Mypm \sigma$ & $\overline{x}\Mypm \sigma$ & $\overline{x}\Mypm \sigma$ &  $\overline{x}\Mypm \sigma$
							
							\\
							
							\hline	 
							Lognormal(0,0.5) & 20 &	20 & 0.1  & 0.472 $\Mypm$   0.279 & 0.477 $\Mypm$   0.287 & 0.470 $\Mypm$   0.285 & 0.473 $\Mypm$   0.286 & 0.483 $\Mypm$   0.287 \\ 
							Lognormal(0,0.5) & 50 &	50 & 0.10  & 0.508 $\Mypm$   0.283 & 0.504 $\Mypm$   0.279 & 0.508 $\Mypm$   0.286 & 0.504 $\Mypm$   0.281 & 0.515 $\Mypm$   0.296 \\ 
							Lognormal(0,0.25) & 20 & 20 & 0.1  & 0.484 $\Mypm$   0.291 & 0.476 $\Mypm$  0.295 & 0.473 $\Mypm$   0.291 & 0.471 $\Mypm$   0.293 & 0.487 $\Mypm$   0.290 \\ 
							Lognormal(0,0.25) & 50 & 50 & 0.1 & 0.517 $\Mypm$   0.292 & 0.523 $\Mypm$  0.291 & 0.522 $\Mypm$   0.293 & 0.522 $\Mypm$   0.291 & 0.506 $\Mypm$   0.284 \\ 
							Lognormal(0,0.0.5) & 20 & 20 & 0.3   &  0.495 $\Mypm$   0.285 & 0.489 $\Mypm$  0.288 & 0.488 $\Mypm$   0.287 & 0.485 $\Mypm$  0.286 & 0.516 $\Mypm$   0.289 \\ 
							Lognormal(0,0.5) & 50 &	50 & 0.3  & 0.476 $\Mypm$   0.296 & 0.473 $\Mypm$   0.293 & 0.468 $\Mypm$   0.287 & 0.47 $\Mypm$   0.292 & 0.487 $\Mypm$   0.296 \\ 
							Lognormal(0,0.25) & 20 & 20 & 0.3  & 0.516 $\Mypm$   0.293 & 0.526 $\Mypm$   0.306 & 0.524 $\Mypm$   0.303 & 0.524 $\Mypm$   0.306 & 0.518 $\Mypm$   0.286 \\ 
							Lognormal(0,0.25) & 50 & 50 & 0.3  & 0.491 $\Mypm$   0.289 & 0.500 $\Mypm$   0.296 & 0.496 $\Mypm$   0.295 & 0.498 $\Mypm$   0.295 & 0.494 $\Mypm$   0.289 \\ 
							\hline

						\end{tabular}}

						\scalebox{0.8}{
							\begin{tabular}{ccccccccc}			
								\hline
								Method: & &	& &	 Energy distance   & Kernel & Kernel  \\
								& &	& &	 $\alpha=1$ & Gaussian $\sigma=1$  & $Laplacian$ $\sigma=1$    \\
								Comparative	& $n_1$ & $n_2$ & Censoring rate   & $\overline{x}\Mypm \sigma$ & $\overline{x}\Mypm \sigma$ &  $\overline{x}\Mypm \sigma$
								\\						
								\hline	 
								Lognormal(0,0.5) & 20 &	20 & 0.1 & 0.490 $\Mypm$   0.287  & 0.490 $\Mypm$   0.283 & 0.493 $\Mypm$   0.287   \\ 
								Lognormal(0,0.5) & 50 &	50 & 0.1 & 0.503 $\Mypm$   0.283 & 0.500 $\Mypm$   0.283 & 0.500 $\Mypm$   0.283    \\ 
								Lognormal(0,0.25) & 20 & 20 & 0.1 & 0.481 $\Mypm$   0.294 & 0.481 $\Mypm$   0.294 & 0.482 $\Mypm$   0.296  \\ 
								Lognormal(0,0.25) & 50 & 50 & 0.1  & 0.517 $\Mypm$   0.289  & 0.517 $\Mypm$   0.291 & 0.516 $\Mypm$   0.288  \\ 
								Lognormal(0,0.0.5) & 20 & 20 & 0.3  & 0.495 $\Mypm$   0.288  & 0.495 $\Mypm$   0.287 & 0.497 $\Mypm$   0.287  \\ 
								Lognormal(0,0.5) & 50 &	50 & 0.3  & 0.482 $\Mypm$   0.293  & 0.482 $\Mypm$   0.294 & 0.488 $\Mypm$   0.297  \\ 
								Lognormal(0,0.25) & 20 & 20 & 0.3 & 0.522 $\Mypm$   0.293  & 0.526 $\Mypm$   0.298  & 0.526 $\Mypm$   0.298  \\ 
								Lognormal(0,0.25) & 50 & 50 & 0.3  & 0.504 $\Mypm$   0.291 & 0.501 $\Mypm$   0.296  & 0.502 $\Mypm$   0.296  \\ 
								\hline		
							\end{tabular}}
						\end{table}

\section{Additional content}
\label{appendix:F}
The Bessel functions of the second order $\Gamma(\cdot)$ (see  Table \ref{tab:tabla1})  are solutions of the Bessel differential equations that have a singularity at $x=0$. Bessel's differential equations are defined as follows:

$$ x^2 \frac{d^2 \Gamma}{dx^2} + x \frac{d\Gamma}{dx} + \left(x^2 - \alpha^2 \right)\Gamma = 0 $$

for an arbitrary complex number $\alpha$,  the order of the Bessel function.

\end{appendices}

\newpage

\bibliographystyle{ims}
\bibliography{citations}

\end{document}